\documentclass[11pt]{article}

\usepackage{fullpage,amsmath,amsfonts,amssymb,amsthm,graphicx,xcolor,hyperref,algorithm,algorithmic,lmodern,enumerate,enumitem,multirow,verbatim,mathabx,caption,subcaption}
\usepackage{mdframed}
\usepackage{xcolor}
\usepackage{mathtools}

\definecolor{darkgreen}{rgb}{0.0, 0.5, 0.0} 

\hypersetup{
    colorlinks=true,
    linkcolor=darkgreen,  
    citecolor=blue,       
    urlcolor=blue,        
    pdfborder={0 0 0}     
}

\newcommand{\ignore}[1]{}

\ignore{
\newtheorem{theorem}{Theorem}[section]
\newtheorem{claim}[theorem]{Claim}
\newtheorem{proposition}[theorem]{Proposition}
\newtheorem{lemma}[theorem]{Lemma}
\newtheorem{corollary}[theorem]{Corollary}
\newtheorem{conjecture}[theorem]{Conjecture}
\newtheorem{observation}[theorem]{Observation}
\newtheorem{fact}[theorem]{Fact}
\theoremstyle{definition}
\newtheorem{example}[theorem]{Example}
\newtheorem{aalgorithm}[theorem]{Algorithm}
\newtheorem{definition}[theorem]{Definition}
\newtheorem{assumption}[theorem]{Assumption}
\newtheorem{remark}[theorem]{Remark}
\newtheorem{problem}[theorem]{Problem}
\newtheorem*{theorem_notag}{}
}
\newtheorem{theorem}{Theorem}[section]

\newtheorem{lemma}[theorem]{Lemma}

\theoremstyle{definition}

\newtheorem{definition}[theorem]{Definition}

\newtheorem{problem}[theorem]{Problem}

\newcommand{\naturals}{\mathbb{N}}
\newcommand{\integers}{\mathbb{Z}}
\newcommand{\reals}{\mathbb{R}}

\newcommand{\coss}[1]{\cos\mkern-2mu\left(#1\right)}
\newcommand{\arccoss}[1]{\arccos\mkern-2mu\left(#1\right)}
\newcommand{\sinn}[1]{\sin\mkern-2mu\left(#1\right)}
\newcommand{\tann}[1]{\tan\mkern-2mu\left(#1\right)}

\renewcommand{\d}[1]{\operatorname{d}\mkern-2mu{#1}}
\newcommand{\norm}[1]{\left\lVert#1\right\rVert}
\def\dd{\mathrm{d}}

\newcommand{\adi}{\textsc{ADI}}

\newcommand{\syst}{\textsc{Sys}}
\newcommand{\spocp}{\textsc{SPOCP}}

\begin{document}

\title{
Optimal Average Disk-Inspection via Fermat’s Principle 
\thanks{This is the full version of a paper with the same title accepted to the 43rd International Symposium on Theoretical Aspects of Computer Science (STACS~2026), Grenoble, March~9--13,~2026.}
}

\author{%
  Konstantinos Georgiou\thanks{Research supported in part by an NSERC Discovery Grant and by the Toronto Metropolitan University Faculty of Science Dean’s Research Fund.}\\[1ex]
  \small Department of Mathematics, Toronto Metropolitan University,Toronto, Ontario, Canada\\
  \small \texttt{konstantinos@torontomu.ca}
}

\maketitle

\begin{abstract}

This work resolves the optimal average-case cost of the Disk-Inspection problem, a variant of Bellman's 1955 lost-in-a-forest problem. In Disk-Inspection, a mobile agent starts at the center of a unit disk and follows a trajectory that inspects perimeter points whenever the disk does not obstruct visibility. The worst-case cost was solved optimally in 1957 by Isbell~\cite{isbell1957optimal}, but the average-case version remained open, with heuristic upper bounds proposed by Gluss~\cite{gluss1961alternative} in 1961 and improved only recently in~\cite{conley2025multi}.  

Our approach applies Fermat's Principle of Least Time to the discretization framework of~\cite{conley2025multi}, showing that optimal solutions are captured by a one-parameter family of recurrences independent of the discretization size. In the continuum limit these recurrences give rise to a single-parameter optimal control problem, whose trajectories coincide with limiting solutions of the original Disk-Inspection problem. A crucial step is proving that the optimal initial condition generates a trajectory that avoids the unit disk, thereby validating the optics formulation and reducing the many-variable optimization to a rigorous one-parameter problem. In particular, this disproves Gluss's conjecture~\cite{gluss1961alternative} that optimal trajectories must touch the disk.  

Our analysis determines the exact optimal average-case inspection cost, equal to $3.549259\ldots$ and certified to at least six digits of accuracy.

\vspace{0.5cm}
\noindent\textit{Keywords:} Inspection, Disk, Average-Case Performance

\end{abstract}

\ignore{
\newpage
\tableofcontents
\newpage
}

\newpage

\tableofcontents

\newpage

\section{Introduction}

A hiker (mobile agent) stands in a forest, knowing only that the boundary is a straight infinite line at distance one but not its orientation. The task is to design a trajectory that guarantees reaching the boundary. Performance is evaluated in the worst-case (adversarial orientation) or in the average-case (orientation drawn uniformly at random).

The above setting is Bellman’s 1955 \emph{lost-in-a-forest} problem specialized to a halfspace with known distance, commonly called the \emph{shoreline problem with known distance}. Isbell~\cite{isbell1957optimal} solved the worst-case version optimally in 1957. The average-case variant was first examined by Gluss~\cite{gluss1961alternative}, who proposed heuristic search strategies to the problem. 
He further conjectured that optimal solutions, as in the worst-case setting, must touch the unit disk. 
More recently, the same problem was studied again from the equivalent perspective of \emph{Disk-Inspection}, leading to a high-dimensional nonconvex NonLinear Program (NLP)~\cite{conley2025multi} modelling a discretized version of the problem, which achieved an average-cost upper bound of $3.5509015$. However, that approach could not certify global optimality of the NLP solutions (which would imply that the reported bound is close to the true optimum), the derived solutions did not scale computationally (which would theoretically give better upper bounds), and it offered no structural characterization of optimal trajectories in the continuous setting. Moreover, no nontrivial lower bound for the average-case cost was known, leaving it entirely unclear how close the reported bound was to the truth.

In this work, we resolve the Average-Case Disk-Inspection problem. We prove that the exact optimum is 
$
3.54925958\ldots,
$
accurate to at least six digits, and that it is realized by trajectories determined through the solutions of an ODE system. Our key technical contribution is to reformulate the discretized problem, via Fermat’s Principle of Least Time, as an optics problem. Assuming an optimal trajectory does not touch the disk, we show that optimal solutions to the discrete inspection problem admit a recurrence structure that, in the continuum limit, yields a single-parameter optimal control problem. In this optimization problem, feasible solutions are trajectories of an ODE system determined by one initial condition, and optimizing over this parameter gives a rigorously certifiable optimum. A crucial part of the analysis is to prove that the non-touching condition is in fact satisfied, so the reduction to the optics problem is valid. Overall, this transforms the previous nonconvex many-variable optimization of~\cite{conley2025multi} with no guarantees into a tractable one-parameter problem, completing the proof of optimality and establishing that an optimal trajectory avoids the unit disk, contrary to Gluss’s conjecture~\cite{gluss1961alternative}, thereby settling the optimal solution to the average-case problem definitively.


\subsection{Related Work}
\label{sec: related work}

In the mid-50s, Bellman~\cite{bellman1956minimization} introduced what is nowadays called \emph{lost-in-a-forest-problem}, one of the earliest formal questions in search theory. The input is a region $R$ (the forest) containing a point $P$ (starting position of a searcher/mobile agent) chosen at random, and the objective is to design a trajectory that minimizes the expected time to reach the boundary, starting from $P$. This formulation became a cornerstone for later work on search under uncertainty, and still many variations that have been proposed through the decades remain open. 

Surveys of the lost-in-a-forest problem are given in~\cite{berzsenyi1995lost,finch2004lost}. Work on specific domains includes convex polygons~\cite{gibbs2016bellman}, strips~\cite{finch2004lost} and general regions analyzed through competitive ratios~\cite{kubel2021approximation}. 
When region $R$ is a half-space, and the starting position of the agent is fixed deterministically, the problem is known as the \emph{shoreline problem} and was first studied for a single agent in~\cite{baeza1988searching} and later extended to several agents in~\cite{baeza1997searching,baeza1995parallel,baezayates1993searching,jez2009two}. 
For one agent, the logarithmic spiral is the best strategy known, with ratio $13.81$~\cite{baeza1988searching,finch2005searching}. The strongest lower bounds known are $6.3972$ in general~\cite{baeza1995parallel} and $12.5385$ for cyclic searches~\cite{langetepe2012searching}. For two agents, a double logarithmic spiral achieves ratio $5.2644$~\cite{baeza1995parallel}. The optimal solutions to the one and two agent search problems are still open. 
For $n \geq 3$, trajectories along rays achieve at most $1/\cos(\pi/n)$~\cite{baeza1995parallel}, and lower bounds matching these values were proven for $n \geq 4$ and for $n=3$ in~\cite{AcharjeeGKS19} and~\cite{dobrev2020improved}, respectively. 

The variation to lost-in-a-forest pertaining to our results is the \emph{shoreline problem with known distance} where the searcher knows her distance to the boundary of the forest. 
For minimizing the worst-case cost, the problem was solved optimally by Isbell~\cite{isbell1957optimal} for one agent, and by Dobrev et al. in~\cite{dobrev2020improved} for two agents, while~\cite{conley2025multi} recently extended the results to the case where the hidden shorelines are tangent to a contiguous portion of the disk, rather than the entire disk.

The average-case version of the shoreline problem with known distance was first studied by Gluss~\cite{gluss1961alternative}, who proposed two heuristic strategies and carried out a rigorous expected-cost analysis for each. He further conjectured that any optimal trajectory, as in the worst-case setting, must touch the disk. The numerical values reported in~\cite{gluss1961alternative} were miscalculated, and as explained in the full version~\cite{conley2025multiagentdiskinspectionArxiv} of the conference paper~\cite{conley2025multi}, the best heuristic bound due to Gluss is $3.63489\ldots$. 
The first systematic approach to the average-case problem was given in the conference version~\cite{conley2025multi} and its full version~\cite{conley2025multiagentdiskinspectionArxiv}. Their method introduced a discretization framework and formulated a nonconvex NonLinear Program (NLP) whose feasible solutions correspond to valid inspection trajectories, with the objective exactly capturing their average cost. Solving this NLP for large discretization parameter $k$ produced a sequence of feasible trajectories and rigorous upper bounds, resulting in a reported value of $3.5509015\ldots$. However, because the NLP is nonconvex, the obtained solutions could not be certified as globally optimal, and its computational hardness restricted $k$ to moderate values, preventing sharper approximations. 
No lower bounds for the average-case cost had been studied or reported, so it remained unclear how close the best previously reported bound was to the true optimum.

Beyond shoreline search and lost-in-a-forest formulations, related problems in search theory have been studied extensively. The field itself is well established, with books and surveys providing systematic accounts of its models and applications~\cite{ahlswede1987search,alpern2013search,alpern2006theory,CGK19search,gal2010search}. Within point search in the plane, Langetepe~\cite{langetepe2010optimality} proved that the logarithmic spiral is optimal for a single robot, while~\cite{georgiou2025spiralsLATIN} recently analyzed the case of multi-speed agents. 
Other studies considered search for a circle in a plane~\cite{gluss1961minimax}, 
multi-agent and grid searches with memory constraints~\cite{emek2015many,Emekicalp2014,fricke2016distributed,LangnerKUW15}, and more recent work investigated searches in geometric terrains~\cite{bouchard2018deterministic,pelc2018reaching} and the role of information in determining search cost~\cite{pelc2018information,pelc2019cost}, to name only a few.

\section{Inspection Problems and New Contributions}
\label{sec: new contributions}

In this section we introduce the inspection problems studied in this work and state our main result, before outlining the technical framework that supports it.

\subsection{Problem Definition and the Main Result}
\label{sec: definition}

We begin by defining inspective trajectories and formalizing both the worst-case and average-case versions of Disk-Inspection, leading up to the main theorem.

\begin{definition}[Inspective Trajectory]
\label{def: inspective curve}
A continuous and piecewise differentiable path in $\reals^2$ (i.e., a curve with two endpoints) is called an \emph{inspective trajectory} if its convex hull contains the unit disk centered at one of its endpoints.
\end{definition}

The requirement that the curve segment be piecewise differentiable stems from the fact that the problem we study concerns arc lengths. This leads to the following problem, also known as the shoreline problem with known distance.

\begin{problem}[Worst-Case Disk-Inspection]
\label{prob: wrst inspection}
Find an inspective trajectory of minimum length. 
\end{problem}

The Worst-Case Disk-Inspection problem was solved in~\cite{isbell1957optimal}, with optimal cost $1+\sqrt3+7\pi/6\approx 6.39724$. Inspective trajectories admit two equivalent descriptions, see~\cite{conley2025multi}. First, they \emph{inspect} every point of the unit disk in the following sense. For any perimeter point $P$ there exists a point $A$ on the trajectory such that $\|\lambda A+(1-\lambda)P\|\geq 1$ for all $\lambda\in[0,1]$ (that is, $A$ sees $P$ without obstruction by the disk). Second, the trajectories intersect every line tangent to the unit disk (i.e. they discover eventually all shorelines). In what follows we use the inspection perspective, which will also define the average-case problem. 

Given an inspective trajectory and a point $P$ on the unit disk, define the inspection time $I_P$ as the length of the trajectory from the origin to the first point that inspects $P$ (equivalently, the time for a unit-speed inspector starting at the origin to see $P$ from outside the disk). The Worst-Case Disk-Inspection, i.e. Problem~\ref{prob: wrst inspection}, asks for the inspective trajectory that minimizes $\sup_P I_P$, where the supremum is over all perimeter points $P$. The average-case version replaces the supremum by expectation.

\begin{problem}[Average-Case Disk-Inspection -- \adi]
Find an inspective trajectory that minimizes $\mathbb{E}_P[I_P]$, where $P$ is chosen uniformly on the perimeter of the disk.
\end{problem}

The cost of \adi\ is defined with respect to an inspector starting at the disk’s center and following an inspective trajectory. One may also allow randomized algorithms that draw a curve from a distribution of inspective trajectories. After the algorithm is fixed, the adversary selects a point $P$, and the algorithm’s performance is the expected inspection time at $P$. If the algorithm first applies a uniform random rotation, then all perimeter points are symmetric, and the adversary’s choice of $P$ has no effect. In this case the cost of a distribution equals the expected inspection time of a randomly chosen point, which is minimized by the deterministic curve achieving the smallest such expectation. Thus the deterministic \adi\ problem already captures the randomized model.

The average-case problem was first studied by Gluss~\cite{gluss1961alternative}, who proposed two heuristic strategies and conjectured that optimal trajectories touch the disk. More recently, a discretization-based nonlinear programming framework was developed~\cite{conley2025multiagentdiskinspectionArxiv,conley2025multi}, yielding the first systematic upper bounds, with best reported value $3.5509015$. However, the nonconvexity of this approach precluded global optimality guarantees, and no lower bounds were known, so the exact optimum remained unresolved.  
We resolve this problem with the following result, establishing the optimal value of \adi\ with high numerical accuracy. In particular, we prove that the previously reported upper bound in~\cite{conley2025multi} was indeed very close to optimal, and we obtain this conclusion through a new continuous framework that replaces the earlier discrete, upper-bound approach. The technical contributions underlying this framework are presented in Section~\ref{sec: new technical contributions}.  
We now state our main result.

\begin{theorem}
\label{thm: main optimal result}
$\lvert C_{\mathrm{opt}} - 3.549259\rvert < 10^{-7}$, where $C_{\mathrm{opt}}$ is the cost for an optimal solution to ADI.
\end{theorem}
All numerical errors are controlled as in Section~\ref{sec: proof of technical lemma}, yielding at least six correct decimal digits, with a discussion on numerical robustness in Section~\ref{sec:numerics}. 
The trajectory certifying Theorem~\ref{thm: main optimal result} is shown in Figure~\ref{fig: optimal} and quantified in the next section.

\subsection{New Technical Contributions}
\label{sec: new technical contributions}

Our second contribution is a continuous framework that replaces the discrete approach of~\cite{conley2025multi}. 
In this framework, candidate inspection trajectories are described, partially, by trajectories generated from an ordinary differential equation (ODE) system. 
The system is governed by a single real parameter, which turns \adi\ into a one-parameter optimal control problem. 
We now introduce the ODE system and the associated trajectories, which form the basis of the technical result leading to Theorem~\ref{thm: main optimal result}.

\begin{definition}[ODE system $\syst(\tau_0)$ for $(\psi,\tau)$]
\label{def: psi tau functions ode and Tau}
For a parameter $\tau_0 \in \reals_{\geq 0}$, let $\psi,\tau:[0,1]\to \reals$ be the unique solution to
\begin{align*}
\psi'(x) &= -2\pi + \frac{\cot \psi(x)}{x}, & \psi(0) &= \tfrac{\pi}{2}, \\
\tfrac{1}{2\pi}\tau'(x) &= \tau(x)\cot\psi(x) - 1, & \tau(0) &= \tau_0 .
\end{align*}
\end{definition}

Lemma~\ref{lem: wellposed} (Appendix~\ref{sec: Proofs Omitted from Section sec: new technical contributions}) shows that $\syst(\tau_0)$ is well defined near $x=0$ and extends uniquely to $[0,1]$.
We later solve this system numerically to identify the trajectories defined below.

\begin{definition}[Curve $\mathcal T^{\tau_0}$]
\label{def: tau trajectory}
Let $(\psi,\tau)$ be the solution to $\syst(\tau_0)$. 
The associated curve is
$$
\mathcal{T}^{\tau_0}:[0,1]\to \reals^2,\quad 
\mathcal{T}^{\tau_0}(x)=\big(\mathcal T^{\tau_0}_1(x), \mathcal T^{\tau_0}_2(x)\big),
$$
where
\begin{align*}
\mathcal T^{\tau_0}_1(x) &= \cos(2\pi x) - \tau(x)\sin(2\pi x), \\
\mathcal T^{\tau_0}_2(x) &= -\sin(2\pi x) - \tau(x)\cos(2\pi x).
\end{align*}
\end{definition}
We refer to $\mathcal T^{\tau_0}$, the solution to ODE $\syst(\tau_0)$, as a curve. 
For suitable $\tau_0>0$, this curve corresponds to a portion of the optimal inspective trajectory certifying Theorem~\ref{thm: main optimal result}. Not every choice of $\tau_0$ yields an inspective trajectory, which is why in the following definition we distinguish some $\tau_0$ that result to curves that do not intersect the unit disk. 

\begin{definition}
\label{def: inspection feasible}
$\tau_0 \in \reals_{\geq 0}$ is called \emph{inspection-feasible} for $\syst(\tau_0)$ if:\\ 
- for all $x\in [0,1]$, we have $\|\mathcal T^{\tau_0}(x)\| > 1$, and \\
- there exists $\xi \in (0,1]$ with $\mathcal T^{\tau_0}_1(\xi)=1$.
The set of all inspection-feasible values is denoted by $\mathcal I$. 
For each $\tau_0 \in \mathcal I$, we denote by $\xi(\tau_0)$ the smallest value of $\xi>0$ that satisfies the second condition, and call it the \emph{deployment parameter} of $\tau_0$.\footnote{The terminology reflects that $\xi$ determines the initial deployment angle $\theta=(1-\xi)\pi$ of the deployment phase, see Figure~\ref{fig: optimal}.}
\end{definition}

Following the standard shoreline framework (see, e.g.,~\cite{conley2025multi}, building on~\cite{gluss1961alternative}), we restrict attention to \adi\ trajectories that start with a
a so-called \emph{deployment phase}. 
This is defined as a line segment connecting the center of the disk to some point $A_\infty$ on the line $x=1$, where the slope $\theta$ of segment $OA_\infty$ will be referred to as the \emph{deployment angle}, and hence $A_\infty=(1,\tan\theta)$ with $\theta\in[0,\pi/2)$.
The reader may consult Figure~\ref{fig: optimal} which shows the optimal inspective trajectory certifying Theorem~\ref{thm: main optimal result}.
The initial deployment phase is followed by the so-called \emph{inspection phase} shown as the green curve
whose other endpoint is $A_0$. In the discrete setting (Figure~\ref{fig: DiscreteTrajectory}), the endpoint is denoted $A_k$, where $k$ is proportional to a discretization parameter that tends to infinity. It is therefore natural to denote the limiting endpoint by $A_\infty$.

\begin{figure}[h!]
    \centering
    \begin{minipage}{5cm}
        \includegraphics[width=\linewidth]{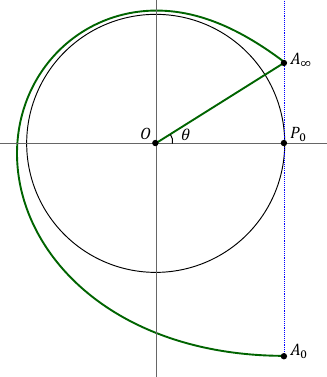}
    \end{minipage}%
    \hspace{1cm}%
    \begin{minipage}{0.58\textwidth}
        \caption{
        The optimal inspective trajectory certifying Theorem~\ref{thm: main optimal result}. 
        The trajectory starts with the deployment phase $OA_\infty$, and is followed by the green curve (inspection phase) $\mathcal T^{\tau_0}$, solution to $\syst(\tau_0)$, for $\tau_0 \approx 1.64697686$, which is inspection-feasible with deployment parameter $\xi(\tau_0)\approx 0.81190987$. 
The inspection phase is therefore the curve $\mathcal T^{\tau_0}(x)$ for $x\in[0,\xi(\tau_0)]$, with endpoints $A_0=\mathcal T^{\tau_0}(0)$ and $A_\infty=\mathcal T^{\tau_0}(\xi(\tau_0))$. 
The dotted blue line marks $x=1$, corresponding to the feasibility condition $\mathcal T^{\tau_0}_1(\xi)=1$. 
The deployment angle is $\theta=(1-\xi(\tau_0))\pi$, with labelled points 
$P_0=(1,0)$, $A_\infty=(1,\tann{\theta})$, and $A_0=(1,-\tau_0)$.
}
        \label{fig: optimal}
    \end{minipage}
\end{figure}

We are now ready to state the main technical result, which expresses the cost of \adi\ as a function of the curve $\mathcal T^{\tau_0}$ and of $(\psi,\tau)$, the solution to $\syst(\tau_0)$. For ease of reference, 
and inspired by standard terminology in control theory,
we call the resulting optimization problem a \emph{Single-Parameter Optimal Control Problem} (\spocp) where each feasible solution is the trajectory of an ODE system determined by a single initial condition serving as the decision variable. In our setting the initial parameter is $\tau_0$, and the corresponding feasible trajectories are generated by $\syst(\tau_0)$. 
We refer to this problem as $\spocp({\tau_0})$, and is formally stated in Theorem~\ref{thm: main result} below.

\begin{theorem}[$\spocp(\tau_0)$ Formulation]
\label{thm: main result}
If the inspection phase of an optimal trajectory for \adi\ does not touch the unit disk, then the inspection phase is the subcurve $\left(\mathcal T^{\tau_0}(x)\right)_{0\le x\le \xi}$ for some inspection-feasible $\tau_0$, where $(\psi,\tau)$ solves $\syst(\tau_0)$, and 
$\xi=\xi(\tau_0)>1/2$ is the deployment parameter of $\tau_0$. 
Together with the deployment segment $O\mathcal T^{\tau_0}(\xi)$, the optimal cost to \adi\ is 
\begin{equation}
\label{equa: function on xi and t0}
\frac{1}{2\pi}\log\!\left(\frac{1+\sinn{\xi \pi}}{1-\sinn{\xi \pi}}\right)
+ \frac{\xi}{\coss{(1-\xi)\pi}}
+ 2\pi \int_0^{\xi} \frac{x \cdot \tau(x)}{\sinn{\psi(x)}} \dd x,
\end{equation}
minimized over all inspection-feasible $\tau_0 \in \mathcal I$, where 
$\theta = (1-\xi)\pi$ is the corresponding deployment angle. 
\end{theorem}

We prove Theorem~\ref{thm: main result} in Section~\ref{sec: optimal control problem}. 
In Section~\ref{sec: proof of technical lemma} we show that the inspection phase of the optimal inspective trajectories do not touch the unit disk. This allows us to invoke Theorem~\ref{thm: main result} and optimize~\eqref{equa: function on xi and t0}, thereby 
solving the underlying one-parameter optimal control problem and establishing Theorem~\ref{thm: main optimal result}.

\section{Background Machinery}
\label{sec: background machinery}

We begin by reviewing the main tools introduced in~\cite{conley2025multi}, which we follow throughout Section~\ref{sec: background machinery} apart from minor reformulations. These ideas play only a preliminary role in our development, since they provide the initial reduction and notation on which our main arguments build. That earlier work proposed a discretized version of \adi\ together with a Nonlinear Programming (NLP) formulation whose solutions yield upper bounds. In Section~\ref{sec: partial average disk inspection} we introduce a related intermediate ``partial'' problem that simplifies the reduction from the continuous setting. Section~\ref{sec: discretize and NLP} then describes the discretized problem and concludes with a brief sketch of the NLP formulation, which, while not directly relevant here, was previously used to derive upper bounds to \adi.

\subsection{Reduction to the Partial Average-Case Disk-Inspection Problem}
\label{sec: partial average disk inspection}

It is convenient to host our arguments on the Cartesian plane and require that the inspecting curves have one endpoint at the origin, where the disk to be inspected is also centered.
More specifically, we parameterize the unit disk perimeter by points
\begin{equation}
\label{equa: pphi def}
P_\phi := \left(
\coss{\phi}, 
\sinn{\phi}
\right),
\end{equation}
where $\phi \in [0,2\pi)$. Therefore, the subject inspective trajectories are functions $T:[0,1]\to \reals^2$, with $T(0)=(0,0)$. 
Towards defining the discretized \adi, we first need to introduce a useful \emph{partial inspection} variant to the problem.

\begin{definition}[$\theta$-Inspective Curves]
\label{def: partial-inspective curve}
Let $\theta \in [0,\pi/2)$. 
A continuous and piece-wise differentiable curve segment in $\reals^2$ (i.e. a curve with two endpoints) is called \emph{$\theta$-inspective} if its convex hull contains a unit disk centered at a point which is $1/\coss{\theta}$ away from one of the curve's endpoints.
\end{definition}

The motivation for introducing $\theta$-inspective curves is that inspection may begin not from the disk center but from a point located at distance $1/\coss{\theta}$ from it. We reserve the term \emph{trajectory} for full solutions to \adi, while \emph{curve} refers to this partial variant. Placing the disk at the origin, we may assume that the starting point is $(1,\tann{\theta})$, which already inspects all boundary points $P_\phi$ with $\phi \in [0,2\theta]$ at zero cost. An example is illustrated in Figure~\ref{fig: optimal}, where the $\theta$-inspective curve has endpoints $A_\infty=(1,\tann{\theta})$ and $A_0$ on the line $x=1$. 

\begin{problem}[$\theta$-Average-Disk-Inspection -- $\theta$-\adi]
\label{prob: theta average cont}
Given $\theta \in [0,\pi/2)$, find a $\theta$-inspective curve that minimizes $\mathbb{E}_{\phi}[I_{P_\phi}]$, where the expectation is over $\phi \in [2\theta,2\pi]$ chosen uniformly at random.
\end{problem}

Any $\theta$-inspective curve can be extended to an inspective trajectory by appending the line segment from the origin to $(1,\tann{\theta})$ (previously referred to as the deployment phase). One of the results in~\cite{conley2025multi} was the following explicit relation between the two problems.

\begin{theorem}
\label{thm: old partial to full}
If $\theta$-\adi\ admits a solution of average cost $s=s(\theta)$, then \adi\ admits a solution of average cost 
$$
B_\theta(s) := \frac{1}{2\pi}
\log \left( \frac{1+\sinn{\theta}}{1-\sinn{\theta}} \right)
+
\left(1 - \frac{\theta}{\pi} \right) 
\left( 
\frac{1}{\coss{\theta}} + s
\right).
$$
\end{theorem}

In this expression, the logarithmic term accounts for the average cost of inspecting $P_\phi$ with $\phi \in [0,2\theta]$ during the initial deployment segment, while the remainder reflects the contribution of the $\theta$-inspective curve. Thus, for each fixed $\theta$, the best partial cost $s_0(\theta)$ yields a full trajectory of cost $B_\theta(s_0)$. It follows that the optimal solution to \adi\ is given by
\begin{equation}
\label{equa: inf wrt theta}
\inf_{\theta \in [0,\pi/2)} B_\theta(s_0).
\end{equation}
The task is therefore reduced to selecting the deployment angle $\theta$ and designing the corresponding $\theta$-inspective curve of minimal average cost $s_0(\theta)$. The starting point for our work is the formulation that lead to the upper bound on $B_\theta(s_0)$ established in~\cite{conley2025multi}.

Theorem~\ref{thm: old partial to full} is the continuous analogue of Lemma~3 in~\cite{conley2025multi}, which is stated for a discrete inspection setting. The proof there is geometric and does not rely on discreteness: the total average cost splits into two independent contributions, namely (i) the average cost of inspecting the initial arc of length \(2\theta\) during the deployment phase, and (ii) the average cost of inspecting the remaining portion of the perimeter, regardless of whether it consists of finitely many or infinitely many points.

Every feasible ADI trajectory admits such a decomposition. Following the classical shoreline analysis of Isbell and Gluss, one may restrict attention to trajectories whose rim begins with a straight deployment segment from the origin to a point \((1,\tan\theta)\) for some \(\theta \in [0,\pi/2)\), and is angularly monotone thereafter. For a fixed \(\theta\), the contribution of the continuation to the total ADI cost depends only on the average cost of the corresponding \(\theta\)-ADI solution. If the continuation were not optimal for that \(\theta\)-ADI instance, replacing it by a better \(\theta\)-ADI curve would strictly improve the total ADI cost. Hence, the continuation of any optimal ADI solution must itself be an optimal solution of the associated \(\theta\)-ADI problem.

\subsection{Disk-Inspection via Discretization (and Nonlinear Programming)}
\label{sec: discretize and NLP}

We now introduce a discretized version of the partial inspection problem and show how it can be modeled as a nonconvex Nonlinear Program (NLP) with $\Theta(k)$ variables, yielding a $(1+1/k)$-approximation to the continuous $\theta$-inspection problem.

Fix $\theta \in [0,\pi/2)$ and $k \geq 5$. Define
\begin{equation}
\label{def:phii}
\phi_i := 2\pi - (\pi - \theta)\tfrac{2i}{k}, 
\quad i = 0,\ldots,k,
\end{equation}
and let $P_{\phi_i}$ denote the corresponding $k+1$ equidistant points on the arc of the unit circle of length $2\pi-2\theta$, see Figure~\ref{fig: DiscreteTrajectory}. For convenience, we will abbreviate $P_{\phi_i}$ by $P_i$ when the index meaning is clear from context.

We assume \(k \ge 5\) so that the angular spacing between consecutive discretization angles is strictly less than \(\pi/2\); this condition is required to ensure that feasibility of the discretized formulation implies feasibility of the corresponding continuous inspection trajectory, as already used in~\cite{conley2025multi}.

\begin{definition}[$(\theta,k)$-Inspective Curves]
\label{def: partial discrete -inspective curve}
A continuous piecewise-differentiable curve segment in $\reals^2$ is called \emph{$(\theta,k)$-inspective} if its convex hull contains all points $P_0,\ldots,P_k$ of a unit disk centered at a point $1/\coss{\theta}$ away from one endpoint of the curve.
\end{definition}

As $k\to\infty$, $(\theta,k)$-inspective curves approximate $\theta$-inspective curves. In particular, scaling a $(\theta,k)$-inspective curve by a factor $1+O(1/k)$ produces a $\theta$-inspective curve, which is one interpretation of the upper bounds derived in~\cite{conley2025multi}.

\begin{problem}[$(\theta,k)$-Average-Disk-Inspection -- $(\theta,k)$-\adi]
\label{prob: theta,k average discrete}
Given $\theta \in [0,\pi/2)$ and $k\geq 5$, find a $(\theta,k)$-inspective curve that minimizes $\mathbb{E}_i[I_{P_i}]$, where $i$ is chosen uniformly at random from $\{0,\ldots,k\}$.
\end{problem}

Let now $L_i(t)$ be the tangent line at $P_i$, parameterized by
\begin{equation}
\label{equa: parametric tangent line}
L_i(t):=
\begin{pmatrix}
\cos(\phi_i) \\
\sin(\phi_i)
\end{pmatrix}
+
t
\begin{pmatrix}
\sin(\phi_i) \\
-\cos(\phi_i)
\end{pmatrix}, \quad t \geq 0.
\end{equation}
Each line $L_i(t)$ passes through $P_i$ at $t=0$ and is tangent to the unit disk. For parameters $t_i \geq 0$, define $A_i := L_i(t_i)$. In particular, we fix $t_k = \tann{\theta}$ so that $A_k=(1,\tann{\theta})$. A candidate $(\theta,k)$-inspective curve is then the polygonal chain $
A_k \to A_{k-1} \to \cdots \to A_0,
$
see Figure~\ref{fig: DiscreteTrajectory}.

The next lemma expresses the average inspection cost in terms of the parameters $t_i$.

\begin{lemma}
\label{lem: cost to discrete}
Let $\theta \in [0,\pi/2)$, $k \geq 5$, and $t=(t_0,\ldots,t_k)\in\reals^{k+1}_{\geq 0}$ with $t_k=\tann{\theta}$. Define $A_i=L_i(t_i)$. Then the polygonal chain $A_kA_{k-1}\cdots A_0$ is $(\theta,k)$-inspective with average inspection cost
\begin{equation}
\label{cost to theta-k inspection}
C_{\theta,k}(t) :=
\frac{1}{k+1}\sum_{i=1}^k i \|A_i - A_{i-1}\|.
\end{equation}
\end{lemma}

\begin{proof}
Points $A_i \in L_i$ inspect $P_i$, so the polygonal chain is $(\theta,k)$-inspective. Each segment $A_{i-1}A_i$ contributes to the inspection time of exactly $i$ points, and averaging over $k+1$ points gives~\eqref{cost to theta-k inspection}.
\end{proof}

It was shown in~\cite{conley2025multi} that if $t_i=\Omega(1/k)$ for all $i$, then the polygonal chain induces a $\theta$-inspective curve for the corresponding continuous instance. Under this feasibility requirement, the discrete average
$C_{\theta,k}(t)$ is a Riemann-sum approximation of the continuous average inspection cost associated with the induced $\theta$-inspective curve, since inspection times vary continuously with the endpoint angle at scale $O(1/k)$. In particular, optimizing $C_{\theta,k}(t)$ over feasible discretizations yields asymptotically tight upper bounds on the continuous $\theta$-\adi\ cost, as used in~\cite{conley2025multi}.

The expression $C_{\theta,k}(t)$ is non-convex in the variables $t$ and $\theta$. Under the feasibility constraints $t_i=\Omega(1/k)$, it serves as the objective of a nonlinear program characterizing the cost of discretized solutions to \adi. In what follows, we show how this program can be reduced to a single-parameter optimization problem and analyze its limiting behavior as $k\to\infty$, resulting in the promised \spocp.

\section{Fermat's Principle Solves $(\theta,k)$-\adi}
\label{sec: fermat gives recursion}

In this section we characterize optimal trajectories to $(\theta,k)$-\adi\ and their cost via a recursion based on Fermat's Principle. 
This marks the beginning of our new contributions. 

\subsection{The Principle of Least Time}
\label{sec: fermat}

We begin by introducing Snell's Law, along with terminology that will be used in subsequent sections. The exposition starts with the \emph{Principle of Least Time}, also known as \emph{Fermat's Principle}, which postulates that the trajectory of a light ray between two given points is the one that minimizes travel time. The principle has been confirmed through experimental observation and explains the rules of refraction in ray optics. 

\begin{figure}[h!]
    \centering
    \begin{minipage}{5cm}
        \includegraphics[width=\linewidth]{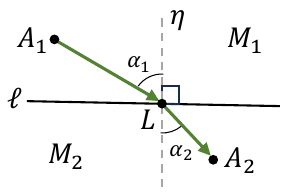}
    \end{minipage}%
    \hspace{1cm}%
    \begin{minipage}{0.55\textwidth}
\caption{
Light refraction across two media $M_1$ and $M_2$, separated by interface line~$\ell$ with normal~$\eta$. A ray from $A_1$ in $M_1$ crosses $\ell$ at $L$ and continues in $M_2$ towards $A_2$, forming angles $\alpha_1,\alpha_2$ with~$\eta$. The same trajectory also applies in reverse, illustrating Snell's Law (Theorem~\ref{thm: snelllaw}).
}
        \label{fig: snelllaw}
    \end{minipage}
\end{figure}

Consider two media with constant speeds $s_1,s_2$, separated by a line $\ell$ with normal $\eta$; see Figure~\ref{fig: snelllaw}. Assume light has constant speed in each medium, hence it travels along a straight ray within each. A ray from $A_1$ to $A_2$ refracts at $L\in\ell$, forming angles $\alpha_1,\alpha_2$ with normal~$\eta$ and obeying Snell's Law below. For simplicity, we refer to phase velocity as speed.

\begin{theorem}[Snell's Law]
\label{thm: snelllaw}
If the speed of light in $M_1,M_2$ is $s_1,s_2$, respectively, then the incidence angle $\alpha_1$ and refraction angle $\alpha_2$ satisfy
\begin{equation*}
\label{equa: snell condition}
\frac{\sinn{\alpha_1}}{\sinn{\alpha_2}} = \frac{s_1}{s_2}. 
\end{equation*}
\end{theorem}

Although Snell's Law is an experimental law of optics, it can be derived rigorously from Fermat's Principle. We present the formal claim next, and we prove it in Appendix~\ref{sec: Proofs Omitted from Section sec: fermat}.

\begin{lemma}
\label{lem: snell optimal}
Let $M_1,M_2$ be two media with constant speeds $s_1,s_2$. Among all continuous paths connecting $A_1\in M_1$ to $A_2\in M_2$, the unique trajectory minimizing travel time is 
piecewise-linear path made of two straight segments $A_1L^\ast,L^\ast A_2$, where $L^\ast\in\ell$ is chosen so that refraction at $L^\ast$ satisfies Snell’s Law.
\end{lemma}

\subsection{Optimal Solution to $(\theta,k)$-\adi\ via Recursion}
\label{sec: optimal sol to theta,k inspection via recursion}

In this section we compute the optimal solution to $(\theta,k)$-\adi\ and its cost using recursion derived from the optics principles of the previous section. 

\begin{figure}[h!]
    \centering
    \begin{minipage}{8cm}
        \includegraphics[width=\linewidth]{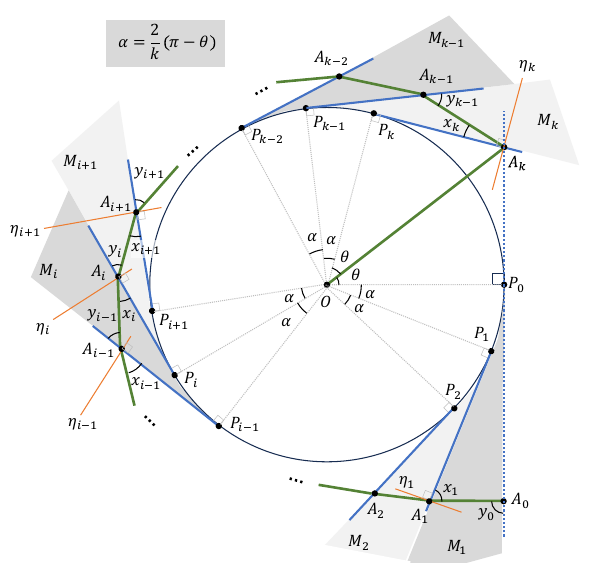}
    \end{minipage}\hspace{1cm}%
    \begin{minipage}{0.4\textwidth}
        \caption{
Geometric setup of the discrete trajectory for $(\theta,k)$-\adi. 
Points $P_i$ on the perimeter, tangent lines $L_i$, and trajectory points $A_i$ define angles $x_i,y_i$ and distances $d_i$. 
The regions $M_i$ represent optical media used in the refraction based interpretation leading to the recursions of Lemma~\ref{lem: recursion for discrete theta}.
}
        \label{fig: DiscreteTrajectory}
    \end{minipage}
\end{figure}

Fix $\theta \in [0,\pi/2)$ and $k\in \integers$, see Figure~\ref{fig: DiscreteTrajectory}. 
A solution to $(\theta,k)$-\adi\ is determined by values $t_i\geq 0$ that specify points $A_i=L_i(t_i)$ on tangent halflines $L_i(t)$, $i=0,\ldots,k$, which in turn inspect perimeter points $P_i$. 
For given $t_i$ we define the following counterclockwise angles:
\begin{align*}
x_i &:= \text{angle formed by } A_iA_{i-1} \text{ and } A_iP_i, \quad i=1,\ldots,k,\\
y_i &:= \text{angle formed by } A_iA_{i+1} \text{ and } L_i(t),~ t\geq t_i, \quad i=0,\ldots,k-1,
\end{align*}
and distances
$$
d_i := \|A_i-A_{i-1}\|, \quad i=1,\ldots,k.
$$
Let also denote
$
\alpha := \tfrac{2(\pi-\theta)}{k},
$
that is, $\alpha$ is the angular distance between two consecutive points $P_i,P_{i+1}$ on the perimeter. 

\begin{lemma}
\label{lem: recursion for discrete theta}
Given $\theta \in [0,\pi/2)$ and $k\geq 5$, suppose an optimal trajectory to $(\theta,k)$-\adi\ is identified by $A_i=L_i(t_i)$ with $t_i \geq \tann{\tfrac{\alpha}{2}}$ (so it does not intersect the unit disk).
\footnote{The condition $t_i \ge \tann{\tfrac{\alpha}{2}}$ is the standard geometric feasibility requirement in the $(\theta,k)$-inspection framework of~\cite{conley2025multi}, ensuring that the discrete trajectory does not intersect the unit disk. We note that in the parameter ranges relevant to our later analysis and numerical evaluations (with $k$ large), the optimized values of $t_i$ are observed to be bounded well away from $0$, so this condition is comfortably satisfied in practice.}
Then the trajectory and its cost are characterized by the recursions
\begin{align}
x_i &= y_{i-1}-\alpha, \label{eq: x_i}\\
y_i &= \arccoss{\tfrac{i}{i+1}\coss{y_{i-1}-\alpha}}, \label{eq: y_i}\\
t_i &= \Big(t_{i-1}-\tann{\tfrac{\alpha}{2}}\Big)\tfrac{\sinn{y_{i-1}}}{\sinn{x_i}}-\tann{\tfrac{\alpha}{2}}, \label{eq: t_i}\\
d_i &= \Big(t_{i-1}-\tann{\tfrac{\alpha}{2}}\Big)\tfrac{\sinn{\alpha}}{\sinn{x_i}}, \label{eq: d_i}
\end{align}
for $i=1,\ldots,k$. The initial conditions are $y_0=\pi/2$ and $t_k=\tann{\theta}$. 
\end{lemma}

\begin{proof}
We work in Cartesian coordinates with the unit disk centered at the origin and tangent to the line $x=1$ at $P_0=(1,0)$. 
Points $P_i$ (see~\eqref{equa: pphi def},~\eqref{def:phii}) and tangent lines $L_i(t)$ (see~\eqref{equa: parametric tangent line}) are as defined earlier. 
By Lemma~\ref{lem: cost to discrete}, a trajectory is given by $A_i=L_i(t_i)$ with $t_i\geq 0$, and the optimal cost is obtained by minimizing $C_{\theta,k}(t)$ over $t\in \reals_{\ge 0}^{k+1}$. 
This yields the path $A_k\to A_{k-1}\to\cdots\to A_0$.

\begin{figure}[h!]
  \centering
  \includegraphics[width=7cm]{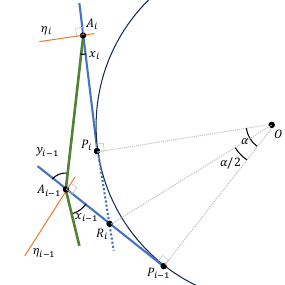}
  \caption{
Close-up detail of Figure~\ref{fig: DiscreteTrajectory} showing point $R_i$ obtained by the intersection of the lines tangent on the disc at points $P_i,P_{i-1}$. 
  }
  \label{fig: DiscreteTrajectory-refined}
\end{figure}

Consider triangle $A_iA_{i-1}R_i$, where $R_i$ is the intersection of the extensions of halflines $L_i$ and $L_{i-1}$ to full lines, see
Figure~\ref{fig: DiscreteTrajectory-refined}. 
By the definition of $t_i$, we have that $\|P_{i-1}A_{i-1}\|=t_{i-1}$. 

Moreover, $R_i$ is the intersection of two tangents at points $P_{i-1},P_i$ whose central angle is $\alpha$. Since $\ell_{i-1}\perp OP_{i-1}$ and $\ell_i\perp OP_i$, the line $OR_i$ bisects $\angle P_{i-1}OP_i$, hence $\angle P_{i-1}OR_i=\alpha/2$. In the right triangle $OP_{i-1}R_i$ we obtain $\|P_{i-1}R_i\|=\tann{\tfrac{\alpha}{2}}$. 
Since the lemma assumes $t_j\ge \tann{\tfrac{\alpha}{2}}$ for all $j$, in particular $t_{i-1}\ge \tann{\tfrac{\alpha}{2}}$, we have that $\|R_iA_{i-1}\|=t_{i-1}-\tann{\tfrac{\alpha}{2}}\ge 0$.
The angle at $A_i$ is $x_i$, the angle at $A_{i-1}$ is $\pi-y_{i-1}$, and the angle at $R_i$ equals $\alpha$ (from quadrilateral $OP_iR_iP_{i-1}$). 
Thus $x_i+(\pi-y_{i-1})+\alpha=\pi$, which gives~\eqref{eq: x_i}.

We interpret the trajectory as an optical path through media $M_i:=L_i^+\cap L_{i-1}^-$ with speeds $s_i:=1/i$, where $L_i^+$ is the halfspace bounded by $L_i$ not containing the disk (other than $P_i$) and $L_i^-$ is the complementary halfspace. 
Each segment $A_{i-1}A_i$ lies in $M_i$, and the total travel time is
$$
\sum_{i=1}^k \frac{\|A_i-A_{i-1}\|}{s_i}
=\sum_{i=1}^k i \|A_i-A_{i-1}\|
=\sum_{i=1}^k i d_i,
$$
which, by Lemma~\ref{lem: cost to discrete}, is proportional to $C_{\theta,k}(t)$ by a factor independent of $t$. Hence, by Lemma~\ref{lem: snell optimal}, refraction holds at each interface.
Indeed, fix $i\in\{1,\ldots,k-1\}$ and consider the interface line associated with $L_i$ separating the two media $M_i$ and $M_{i+1}$. If we keep the points $A_{i-1}\in L_{i-1}$ and $A_{i+1}\in L_{i+1}$ fixed, then any feasible trajectory induces a unique crossing point $A_i\in L_i$, and the contribution of the two adjacent segments to the travel time equals
$$
i\|A_i-A_{i-1}\|+(i+1)\|A_{i+1}-A_i\|.
$$
All other terms in the total travel time are independent of the choice of $A_i$. Since the trajectory is optimal, the chosen point $A_i$ must minimize the above expression over $A_i\in L_i$. Therefore Lemma~\ref{lem: snell optimal} applies at $L_i$, and the segment directions satisfy Snell's law between $M_i$ and $M_{i+1}$. If refraction failed at some interface, replacing $A_i$ by the minimizer given by Lemma~\ref{lem: snell optimal} would strictly decrease the total travel time, contradicting optimality.

At the interface between $M_i$ ($s_i=1/i$) and $M_{i+1}$ ($s_{i+1}=1/(i+1)$), the incidence and refraction angles with respect to the normal are
$
\alpha_1=\tfrac{\pi}{2}-x_i$ and $\alpha_2=\tfrac{\pi}{2}-y_i$. 
Therefore, by Snell’s Law (Theorem~\ref{thm: snelllaw}), we have 
$
\tfrac{\sinn{\alpha_1}}{\sinn{\alpha_2}}=\tfrac{s_i}{s_{i+1}}=\tfrac{i+1}{i}.
$
Since $\sinn{\tfrac{\pi}{2}-z}=\coss{z}$, this becomes
$
\tfrac{\coss{x_i}}{\coss{y_i}}=\tfrac{i+1}{i}
$,
which together with~\eqref{eq: x_i} yields~\eqref{eq: y_i} .

In triangle $A_iA_{i-1}R_i$ we have $\|R_iP_i\|=\|P_{i-1}R_i\|=\tann{\tfrac{\alpha}{2}}$ and $t_i=\|P_iA_i\|$, hence
$$
\|R_iA_i\|=t_i+\tann{\tfrac{\alpha}{2}},\qquad
\|R_iA_{i-1}\|=t_{i-1}-\tann{\tfrac{\alpha}{2}}.
$$
The angle at $A_{i-1}$ equals $\pi-y_{i-1}$, so $\sinn{\pi-y_{i-1}}=\sinn{y_{i-1}}$. By the Sine Law,
$$
\frac{\sinn{y_{i-1}}}{ t_i+\tann{\tfrac{\alpha}{2}} }
=
\frac{\sinn{x_i}}{ t_{i-1}-\tann{\tfrac{\alpha}{2}} }
=
\frac{\sinn{\alpha}}{ d_i },
$$
which implies~\eqref{eq: t_i} and~\eqref{eq: d_i}.

For the initial conditions, $t_k=\tann{\theta}$ holds by construction. 
At the other end, optimality requires that the trajectory satisfies both the refraction rule at every interface and, for the final segment, minimize $\|A_0-A_1\|$ over $A_0\in L_0$. 
This minimization places $A_0$ at the orthogonal projection of $A_1$ onto $L_0$, so the angle between $A_0A_1$ and $L_0$ equals $\pi/2$, that is $y_0=\pi/2$.
\end{proof}

\section{Single-Parameter Reformulation and Continuum Limit}
\label{sec: Single-Parameter Reformulation and Continuum Limit}

The backbone of the argument relies on~\eqref{equa: inf wrt theta}, which states that the optimal cost to \adi\ can be computed as $\inf_{\theta \in [0,\pi/2)} B_\theta(s)$, where $s=s(\theta)$ is the optimal solution to $\theta$-\adi. 
We show that $B_\theta(s)$ can, under suitable conditions, be minimized as a \spocp\ in  parameter $\tau_0$. Thus our first task is to verify that these conditions hold for the deployment parameter $\theta$ and the corresponding $\tau_0$ that generate solutions to $\syst(\tau_0)$, and then to make the relation between $\theta$ and $\tau_0$ explicit. 

Section~\ref{sec: optimal control problem} formulates the minimization of $B_\theta(s(\theta))$ as a $\spocp(\tau_0)$. 
Section~\ref{sec: Restriction of Deployment Angles} restricts the range of deployment angles $\theta$ so that the conditions apply. 
Section~\ref{sec: proof of technical lemma} relates $\theta$ and $\tau_0$ and solves the resulting $\spocp(\tau_0)$ numerically. 
Section~\ref{sec:numerics} discusses the robustness of these numerical results.

\subsection{A Single Parameter Optimal Control Problem -- Proof of Theorem~\ref{thm: main result}}
\label{sec: optimal control problem}

We now prove Theorem~\ref{thm: main result}. The starting point is the recurrence of Lemma~\ref{lem: recursion for discrete theta}, which describes the optimal inspective curve as an optics trajectory with refraction angles independent of $t_k=\tann{\theta}$. The deployment angle $\theta$ still determines the trajectory and its cost, and fixes the endpoint $A_0=L_0(t_0)=(1,-t_0)$. Thus the trajectory connecting two points of the line $x=1$ in the first and fourth quadrants (see Figure~\ref{fig: DiscreteTrajectory}) may be viewed either as a ray starting from $A_k$ in the first quadrant or from $A_0$ in the fourth.

For  technical reasons, rather than parameterizing the recurrence by $t_k=\tann{\theta}$, we work with $t_0=\tau_0$, restricted to the feasible range of Definition~\ref{def: inspection feasible}. The challenge is then to determine the point where the trajectory intersects the line $x=1$ again, and hence the corresponding deployment angle $\theta$. In Lemma~\ref{lem: recursion for discrete theta} the angular step is $\alpha=2(\pi-\theta)/k$, still vanishing as $k\to\infty$ but expressed in terms of $\theta$. To eliminate this dependence we place $n$ equispaced points $P_0,\ldots,P_n$ on the unit circle, with step $\alpha=2\pi/n$, and recover $\theta$ as a function of $\tau_0$.

This setup leads to the following continuum limit as $n\to\infty$, namely the ODE system $\syst(\tau_0)$ of Definition~\ref{def: psi tau functions ode and Tau}. The corresponding trajectory $\mathcal T$ begins at $(1,-\tau_0)$, remains outside the unit disk, and intersects the line $x=1$ in the first quadrant at some point $\mathcal T(\xi)$, where $\xi$ is the deployment parameter.
We now state a sequence of technical results establishing this limit.
We defer their proofs to Appendix~\ref{sec: proofs omitted from sec: optimal control problem}.

First, we show that if the discrete recursion is extended to continuous functions by connecting consecutive values linearly, the resulting piecewise-linear functions converge to the solution of the ODE system $\syst(\tau_0)$ of Definition~\ref{def: psi tau functions ode and Tau}. 

\begin{lemma}
\label{lem: continuum-limit}
Let $(y_i)_{i=0}^n$ and $(t_i)_{i=0}^n$ be defined by \eqref{eq: x_i}--\eqref{eq: t_i} with $\alpha=2\pi/n$, $y_0=\pi/2$, and $t_0=\tau_0$. Define piecewise-linear interpolants $\psi_n,\tau_n:[0,1]\to\reals$ by $\psi_n(i/n)=y_i$ and $\tau_n(i/n)=t_i$, extended linearly on each $[i/n,(i+1)/n]$. Then $(\psi_n,\tau_n)$ converges uniformly on compact subsets of $(0,1]$ to $(\psi,\tau)$, the unique solution to the ODE system $\syst(\tau_0)$ with $\psi(0)=\pi/2$ and $\tau(0)=\tau_0$.
\end{lemma}

We are now ready to provide the limiting behavior of the inspecting curve, giving rise to $\mathcal T^{\tau_0}$ as in Definition~\ref{def: tau trajectory}. In the remainder of the section, for notational convenience, we drop the superscript and write simply $\mathcal T$. The next lemma
shows that the polygonal trajectories converge indeed to $\mathcal T$.

\begin{lemma}
\label{lem: trajectory-representation}
Let $\phi_i$ and $L_i(t)$ be defined as in \eqref{def:phii} and \eqref{equa: parametric tangent line}, and let $A_i=L_i(t_i)$, where $(t_i)$ is defined by \eqref{eq: t_i} with $t_0=\tau_0$.
Assume that $\tau_0$ is inspection-feasible (Definition~\ref{def: inspection feasible}), and let $\xi=\xi(\tau_0)$ be its deployment parameter, that is, the smallest $\xi\in(0,1]$ such that $\mathcal T_1(\xi)=1$ for $\mathcal T=\mathcal T^{\tau_0}$ of Definition~\ref{def: tau trajectory}.
Let $\widetilde{\mathcal T}_n$ be the polygonal path obtained by linear interpolation through the points $(A_i)_{i=0}^n$ with parameter values $h_i=i/n$.

Then $\widetilde{\mathcal T}_n$ converges to $\mathcal T$ on $[0,\xi)$, in the sense that for every $0<\delta\le \xi'<\xi$,
$$
\sup_{x\in[\delta,\xi']}\lVert \widetilde{\mathcal T}_n(x)-\mathcal T(x)\rVert \to 0
\qquad\text{as } n\to\infty.
$$
Moreover, $\mathcal T(0)=A_0$.
\end{lemma}

Since $\mathcal T$ arises as the continuum limit of the discrete trajectories, the endpoint $A_k$ of the polygonal path (see Figure~\ref{fig: DiscreteTrajectory}) corresponds to an index $k=\Theta(n)$. As $n\to\infty$ this endpoint is denoted $A_\infty$ in Figure~\ref{fig: optimal}. We can now justify the definition of the deployment parameter. 

\begin{lemma}
\label{lem: deployment parameter}
Let $\tau_0$ be inspection-feasible to the ODE system of Definition~\ref{def: psi tau functions ode and Tau}. Then its deployment parameter $\xi=\xi(\tau_0)$ satisfies $\mathcal T_2(\xi)=\tann{\theta}$, where $\theta=(1-\xi)\pi$.
\end{lemma}

From the progress above, and starting with inspection-feasible $\tau_0$, we compute $\theta$-inspective curve to $\theta$-\adi, where $\theta=\theta(\xi)$, and $\xi=\xi(\tau_0)$. The next lemma uses again the continuum construction to derive the cost of $\mathcal T$ to $\theta$-\adi, as a function of $\tau_0$, and in terms of the solution to ODE system $\syst(\tau_0)$.

\begin{lemma}
\label{lem: sol to cont partial}
Let $\tau_0$ be inspection-feasible with deployment parameter $\xi=\xi(\tau_0)$. Then $\mathcal T$ is a feasible solution to $\theta$-\adi, where $\theta=(1-\xi)\pi$, and the average cost equals
$$
\frac{2\pi}{\xi}\int_0^{\xi}\frac{x\tau(x)}{\sinn{\psi(x)}}\dd x.
$$
\end{lemma}

We can now conclude the proof of Theorem~\ref{thm: main result}. By~\eqref{equa: inf wrt theta} and the discussion of Section~\ref{sec: partial average disk inspection}, the optimal inspective curve has cost $\inf_{\theta\in[0,\pi/2)} B_\theta(s_0)$, where $s_0=s_0(\theta)$ is the cost to some $\theta$-\adi\ instance. Since the optimal curve does not touch the disk, the deployment angle is $\theta=(1-\xi)\pi$ for some inspection-feasible $\tau_0$. Lemma~\ref{lem: sol to cont partial} gives $s_0$ as a function of $\xi$, which can then be substituted into Theorem~\ref{thm: old partial to full} to yield the expression of Theorem~\ref{thm: main result}.

Finally, since $\theta\in[0,\pi/2)$ we have $\xi>1/2$, ensuring that expression~\eqref{equa: function on xi and t0} is well defined. To conclude, it remains to show that the expression attains a minimum, which follows from continuity together with the Extreme Value Theorem.

\begin{lemma}
\label{lem: int attains min}
Fix an inspection-feasible initial value $\tau_0$ with deployment parameter $\xi\in(1/2,1]$. Define
$
I(\xi):=2\pi\int_0^{\xi}\tfrac{x\tau(x)}{\sinn{\psi(x)}}\dd x.
$
Then $I(\xi)$ is well defined and belongs to $C^1([0,\xi])$, with
$
I'(\xi)=2\pi\xi\tfrac{\tau(\xi)}{\sinn{\psi(\xi)}}.
$
\end{lemma}

\subsection{Bounds on the Optimal Deployment Angle}
\label{sec: Restriction of Deployment Angles}

In light of Theorem~\ref{thm: main result} proved in the previous section, the natural approach to finding the optimal solution to \adi\ is to solve the underlying single-parameter (here $\tau_0$) optimal control problem with objective~\eqref{equa: function on xi and t0}. Put differently, we are back to determining the optimal solution as
$\inf_{\theta \in [0,\pi/2)} B_\theta(s_0)$, where $s_0=s_0(\theta)$ is the optimal solution to $\theta$-\adi, see~\eqref{equa: inf wrt theta}. However, not all deployment angles $\theta$ give rise to optimal trajectories that avoid intersecting the disk, which is a premise of Theorem~\ref{thm: main result}. 

For this reason, we restrict the range $[0,\pi/2)$ of deployment angles. Small values of $\theta$ are excluded because the corresponding optimal trajectories touch the disk. Large values of $\theta$ are excluded for a different reason: as $\theta\to\pi/2$, the deployment point moves arbitrarily far from the origin, and we show in the next section that such angles cannot be optimal. Consequently, these angles need not be considered in the optimization.
We then search numerically for the minimizing deployment angle over the remaining range. Restricting to this range also avoids numerical instability in the associated \spocp, which would otherwise arise from the large scales involved. The numerical procedures used in this restricted regime are justified in Section~\ref{sec:numerics}.

To summarize, in this section we show a refinement of~\eqref{equa: inf wrt theta}, as follows.

\begin{lemma}
\label{lem: restriction of thetas}
The optimal solution to \adi\ is given by
$
\inf_{\theta \in [0.52, 1.148)} B_\theta(s_0).
$
\end{lemma}

We start by showing that the deployment angle cannot be too large. 

\begin{lemma}
\label{lem: theta no big}
$\theta_0 \leq 1.148$ for the optimal deployment angle $\theta_0$ to \adi. 
\end{lemma}

\ignore{
test[x_] := (1/(2  Pi))  Log[(1 + Sin[x])/(1 - Sin[x])] + (1 - 
     x/Pi)  (1/Cos[x] + (
     3 \[Pi] + \[Pi]^2 - 2  \[Pi]  x + \[Pi]  Tan[x])/(4  (\[Pi] - x))
     )

FindRoot[ test[x] == 3.5509015, {x, 1}]
     
FullSimplify [ 
 D[ test[x], x], Assumptions -> 0 <= x < Pi/2
 ]     
}

\begin{proof}
Set $\mathcal O = \inf_{\theta \in [0,\pi/2)} B_\theta(s)$. 
The main contribution to~\cite{conley2025multi} is an upper bound of $3.5509015$ to the optimal cost to \adi, that is $\mathcal O \leq 3.5509015$. Next we show that $B_\theta(s) > \mathcal O$ when $\theta >1.148$. 

First, for any fixed $\theta$, we provide a lower bound to the cost of $\theta$-\adi\ in which we are inspecting points $P_\phi$ with $\phi \in [2\theta,2\pi]$. We do this by lower bounding the inspection cost of points $P_\phi$ in the  intervals
$$
I_1 = [2\theta,\pi], \quad
I_2 = [\pi,3\pi/2], \quad
I_3 = [3\pi/2,2\pi].
$$
The inspection cost for points in $I_1$ is $0$. For points in $I_2$, the cost is at least the time required to inspect $P_\pi$, namely $2$. For points in $I_3$, the cost is at least the time required to inspect $P_{3\pi/2}$. In this case we employ the provably optimal trajectory from~\cite{conley2025multi}, established for the worst-case problem. Since all preceding points must already have been inspected by moving counterclockwise around the disk, $P_{3\pi/2}$ cannot be inspected earlier than 
$
\tann{\theta}+(\pi-2\theta)+1,
$
where $\tann{\theta}$ is the tangent length to the disk, $\pi-2\theta$ is the circular arc length, and the additional $1$ is the final straight segment needed to complete visibility without intersecting the disk. Overall, this shows that
$$
s(\theta) \geq 
\tfrac{3\pi/2-\pi}{2\pi-2\theta}\cdot 2
+
\tfrac{2\pi-3\pi/2}{2\pi-2\theta}\cdot \left(\tann{\theta}+(\pi-2\theta)+1\right)
=
\frac{\pi  (\tan (\theta )+\pi -2 \theta + 3)}{4 (\pi -\theta )}.
$$
By Theorem~\ref{thm: old partial to full}, $B_{\theta}(s)$ is increasing in $s$, and therefore
$$
B_{\theta}(s) \geq 
 \frac{1}{2\pi}
\log \!\left( \frac{1+\sinn{\theta}}{1-\sinn{\theta}} \right)
+
\left(1 - \frac{\theta}{\pi} \right) 
\left( 
\frac{1}{\coss{\theta}} 
+\frac{\pi  (\tan (\theta )+\pi -2 \theta + 3)}{4 (\pi -\theta )}
\right).
$$
Call this last expression $h(\theta)$. A direct calculation gives
$$
h'(\theta)=\tfrac{1}{4}(\tan^2\theta-1)+\tfrac{(\pi-\theta)\tan\theta\sec\theta}{\pi}.
$$
One verifies that $h'(\theta)>0$ for $\theta\in[1.148,\pi/2)$, hence $h(\theta)$ is increasing on this interval. Therefore
$$
B_{\theta}(s)\ge h(\theta)\ge h(1.148)\approx 3.55348>\mathcal O,
$$
where $h(1.148)$ is evaluated with numerical precision to at least ten digits, and the reported difference to $\mathcal O$ exceeds the fourth decimal place. Details on the robustness of this computation are deferred to Section~\ref{sec:numerics}.
\end{proof}

\begin{lemma}
\label{lem: cont at least some discrete}
Fix $\theta \in [0,\pi/2)$ and $k\geq 5$. 
Then, the optimal solution to $\theta$-\adi\ is at least the minimum of
$
\sum_{i=1}^{k}  \tfrac{i-1}k  \|A_{i-1}A_{i}\|
$
over all points $A_i$ that form $(\theta,k)$-inspective curves. 
\end{lemma}

\begin{proof}
Consider an optimal trajectory for the partial problem $\theta$-\adi\ which inspects all points $\{P_{\phi}\}_{\phi \in [2\theta,2\pi]}$, each with inspection time $I(P_\phi)$. It follows that such a trajectory is also feasible to $(\theta,k)$-\adi\ problem that inspects $k$ equidistant points in the same arc. 
Note that any point $P_{\phi}$ in interval $\phi \in [\phi_i,\phi_{i-1}]$ is inspected, by the triangle inequality, in time at least 
$
\sum_{j=i+1}^k \|A_{j}A_{j-1}\|.
$
It follows that for the cost to the continuous problem we have 
\begin{align*}
\frac{1}{2\pi-2\theta} 
\int_{2\theta}^{2\pi} I(P_\phi) \, d\phi
&= 
\frac{1}{2\pi-2\theta} 
\sum_{i=1}^{k} \int_{\phi_{i}}^{\phi_{i-1}} I( P_\phi) \, d\phi \\
&\geq  
\frac{1}{2\pi-2\theta} 
\sum_{i=1}^{k} \int_{\phi_{i}}^{\phi_{i-1}} \sum_{j=i+1}^k \|A_{j} A_{j-1}\|  \, d\phi \\
&= 
\frac{1}{2\pi-2\theta} 
\sum_{i=1}^{k-1} \left( \phi_{i-1} - \phi_{i} \right) \sum_{j=i+1}^k \|A_{j}A_{j-1}\|   \\
&= 
\frac{1}{2\pi-2\theta} 
\sum_{i=1}^{k} \left( \phi_{k-i} - \phi_{k-1} \right) \|A_{i-1}A_{i}\|   \\
&= 
\sum_{i=1}^{k}  \tfrac{i-1}k  \|A_{i-1}A_{i}\|.
\end{align*}
Here we used the explicit formula for the equidistant angles $\phi_i$ from~\eqref{def:phii}, which gives $\phi_{k-i}-\phi_{k-1} = \tfrac{2(\pi-\theta)}{k}(i-1)$.
\end{proof}

We are ready to show that the deployment angle cannot be too small. 
\begin{lemma}
\label{lem: theta no small}
$\theta_0 \geq 0.52$ for the optimal deployment angle $\theta_0$ to \adi. 
\end{lemma}

\begin{proof}
We show that deployment angles $\theta_0 < 0.52$ result in solutions to \adi\ of cost strictly more than $3.5509015$, which is a known upper bound. 

For this, we invoke Theorem~\ref{thm: old partial to full}, which expresses the cost $B_\theta(s)$ to \adi\ given deployment angle $\theta$, where $s=s(\theta)$ is the optimal cost to $\theta$-\adi. 
At the same time, Lemma~\ref{lem: cont at least some discrete} provides a lower bound to $s(\theta)$ via $(\theta,k)$-inspective curves identified by points $A_i=L_i(t_i)$, where $t_i\geq 0$, $i=0,\ldots,k$. 

In other words, a lower bound to $s(\theta)$, and subsequently to $B_\theta(s)$, for each $\theta$, can be obtained by solving the nonlinear program
$$
\min\;\sum_{i=1}^{k}  \tfrac{i-1}{k}\, \|A_{i-1}A_{i}\|
\quad \text{subject to } t_i\geq 0,\; i=0,\ldots,k,
$$
where $A_i=L_i(t_i)$ and $\|A_{i-1}A_i\|=\|g_i(t_{i-1},t_i)\|$ for an affine map $g_i$ in the pair $(t_{i-1},t_i)$.

The feasible region $\{t\in\mathbb{R}^{k+1}: t_i\ge 0\}$ is convex. For each $i$, the map $(t_{i-1},t_i)\mapsto \|g_i(t_{i-1},t_i)\|$ is convex, because it is the composition of a convex norm with an affine transformation. Multiplying by the nonnegative weight $(i-1)/k$ preserves convexity, and summing over $i$ preserves convexity. Hence the program is convex, and any local minimum is globally optimal. This justifies that the value we compute is the true discrete lower bound for each fixed $\theta$ and $k$. Details of the numerical implementation and accuracy guarantees are deferred to Section~\ref{sec:numerics}.

\begin{figure}[h!]
    \centering
    \begin{subfigure}{0.32\textwidth}
        \includegraphics[width=\linewidth]{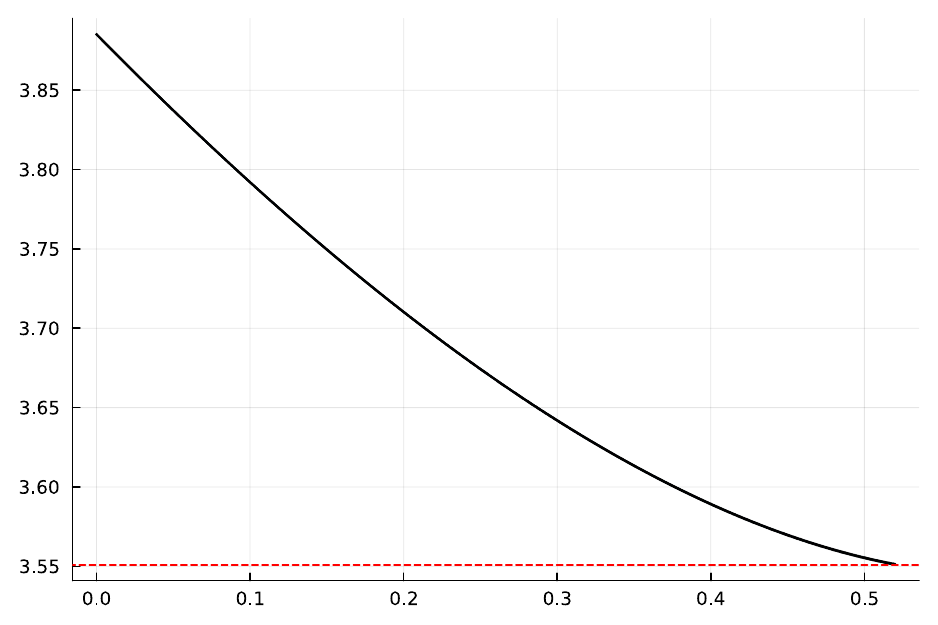}
        \caption{$\theta \in [0 , 0.52]$}
        \label{fig:0to052}
    \end{subfigure}
    \hfill
    \begin{subfigure}{0.32\textwidth}
        \includegraphics[width=\linewidth]{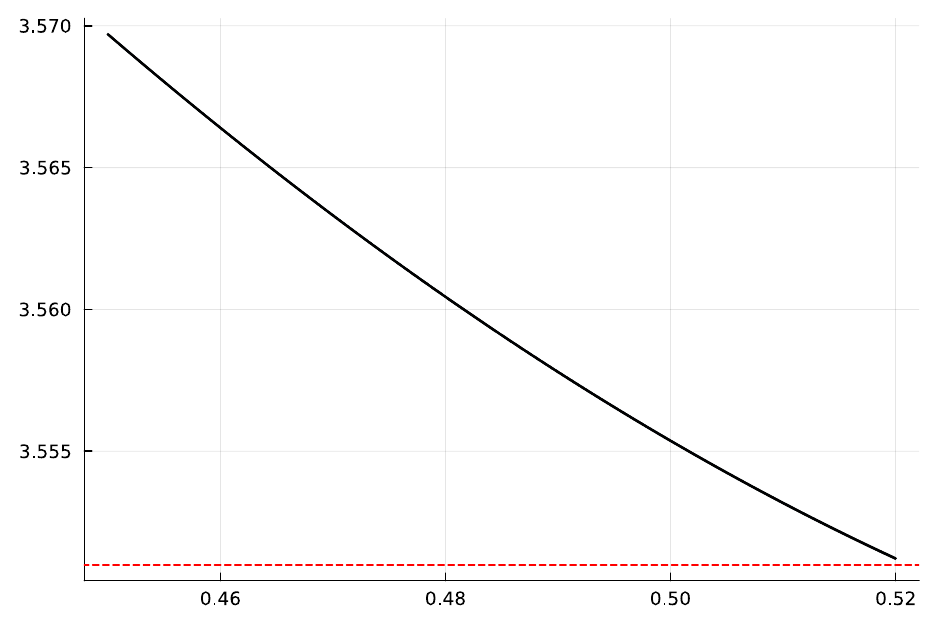}
        \caption{$\theta \in [0.45 , 0.52]$}
        \label{fig:045to052}
    \end{subfigure}
    \hfill
    \begin{subfigure}{0.32\textwidth}
        \includegraphics[width=\linewidth]{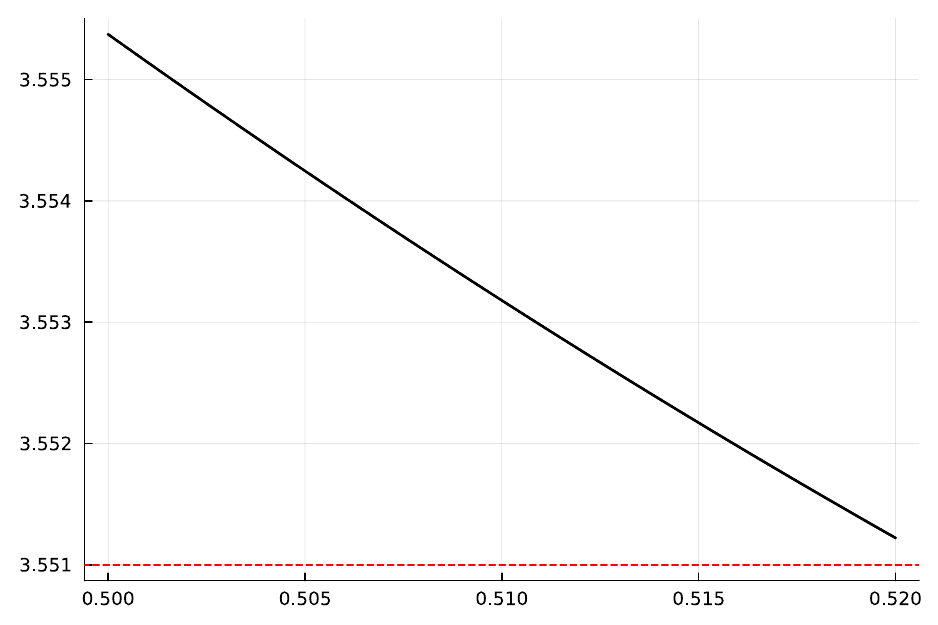}
        \caption{$\theta \in [0.5 , 0.52]$}
        \label{fig:05to052}
    \end{subfigure}
    \caption{The horizontal axis corresponds to $\theta$ and the vertical axis to the numerically calculated lower bound to $B_\theta(s)$. Each plot was computed by a grid of 1000 points in the range of $\theta$'s. The red lines correspond to the line $y=3.551$.}
    \label{fig: num lower bound NLP}
\end{figure}

Figure~\ref{fig: num lower bound NLP} illustrates the numerical values obtained. We performed computations with $k=1000$, evaluating the lower bound on a uniform grid of $1000$ points in the interval $[0,0.52]$. The minimum value observed was $3.5512215\dots$, attained at $\theta=0.52$, and the values form a strictly decreasing sequence throughout this interval. Empirically the sampled values remain above $3.551$ with a margin of about $3.2\times 10^{-4}$ at $\theta=0.52$, and refinements near $[0.5,0.52]$ (Figure~\ref{fig:05to052}) confirm this margin. Therefore, for all $\theta\in[0,0.52]$, the cost exceeds the known upper bound.
\end{proof}

Note that Lemma~\ref{lem: restriction of thetas} is directly implied by Lemmata~\ref{lem: theta no big} and~\ref{lem: theta no small}.

\subsection{The Optimal Solution to \adi\ -- Proof of Theorem~\ref{thm: main optimal result}}
\label{sec: proof of technical lemma}

In this section, we show that the premise of Theorem~\ref{thm: main result} is satisfied, namely that the optimal inspective curve does not touch the unit disk. Consequently, this allows us to optimize expression~\eqref{equa: function on xi and t0} and thus conclude with the proof of Theorem~\ref{thm: main optimal result}. We rely on the theoretical foundations established previously, and, as is necessary, we make use of further numerical calculations implemented in the \texttt{Julia} programming language~\cite{bezanson2017julia}. Many of the arguments below are based on numerical comparisons. In Section~\ref{sec:numerics} we provide the justifications that these computations are robust and consistent with the accuracy promised in the main theorem.

The next lemma analyzes inspection trajectories in a carefully chosen regime of initial values $\tau_0$ to the ODE system $\syst(\tau_0)$. 

\begin{lemma}
\label{lem: tau ok and theta ok}
Every $\tau_0 \in [1.64697 , 1.6525]$ is inspection feasible, and the corresponding deployment parameters $\xi=\xi(\tau_0)$ determine deployment angles $\theta = (1-\xi)\pi$ whose range covers interval $[0.52, 1.148]$.
\end{lemma}

\begin{proof}
We solve the ODE system $\mathcal T^{\tau_0}$ of Definition~\ref{def: psi tau functions ode and Tau} on a grid of 2000 sample points for initial conditions $\tau_0 \in [1.64697 , 1.6525]$. The resulting values are summarized in Figure~\ref{fig: par to ODE system sol}. Recall that the initial condition $\tau(0)=\tau_0$ determines $\xi$, $\theta$ and the entire curve $\mathcal T(\cdot)$.

Figure~\ref{fig:xifortau} reports $\xi=\xi(\tau_0)$, the parameter with $\mathcal T(0)=A_0$ and $\mathcal T(\xi)=A_\infty$. In Figure~\ref{fig:mintfortau}, we compute the minimum of $\tau(x)$ over $x\in [0,\xi]$, which is bounded below by $0.2$. Since $\|\mathcal T(x)\|=\sqrt{1+\tau(x)^2}$, the distance of $\mathcal T(x)$ from the unit circle is $\sqrt{1+\tau(x)^2}-1$. Therefore a uniform bound $\tau(x)\ge 0.2$ implies a radial clearance of at least
$$
\sqrt{1+0.2^2}-1 \approx 0.01980198\ldots .
$$
This proves inspection feasibility for every $\tau_0$ in the stated range. Accuracy of the ODE integration and stability checks are deferred to Section~\ref{sec:numerics}.

Finally, Figure~\ref{fig:thetafortau} shows the corresponding deployment angles $\theta = (1-\xi)\pi$. The endpoints satisfy
$$
\theta(1.64697)\approx 0.501177\ldots \quad\text{and}\quad \theta(1.6525)\approx 1.1600947\ldots,
$$
and the image of $[1.64697,1.6525]$ under $\theta(\cdot)$ contains $[0.52,1.148]$, as indicated by the horizontal reference lines in the figure.
\end{proof}

\begin{figure}[h!]
    \centering
    \begin{subfigure}{0.32\textwidth}
        \includegraphics[width=\linewidth]{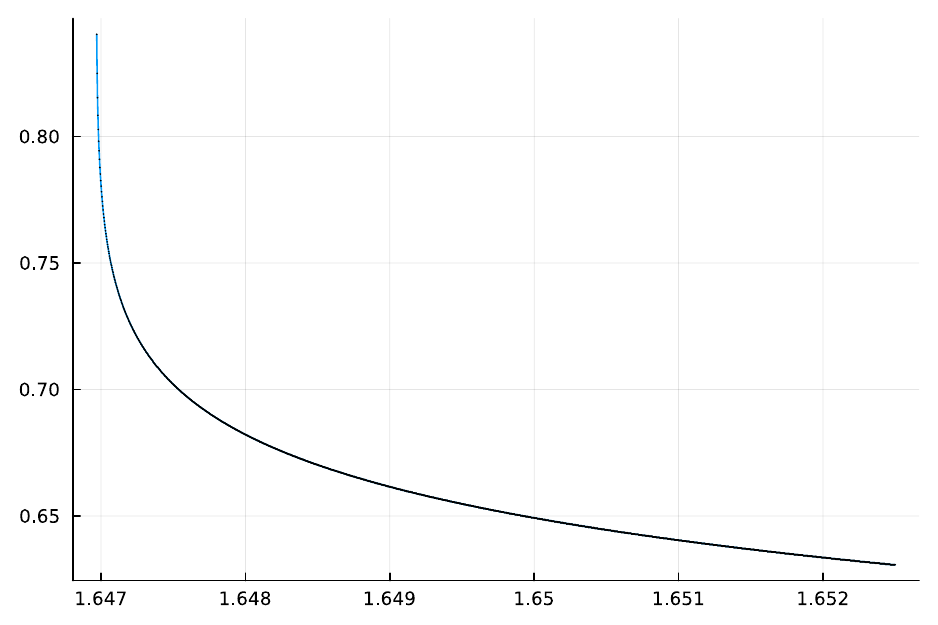}
        \caption{Plot of the deployment parameter $\xi=\xi(\tau_0)$ against $\tau_0$.}
        \label{fig:xifortau}
    \end{subfigure}
    \hfill
    \begin{subfigure}{0.32\textwidth}
        \includegraphics[width=\linewidth]{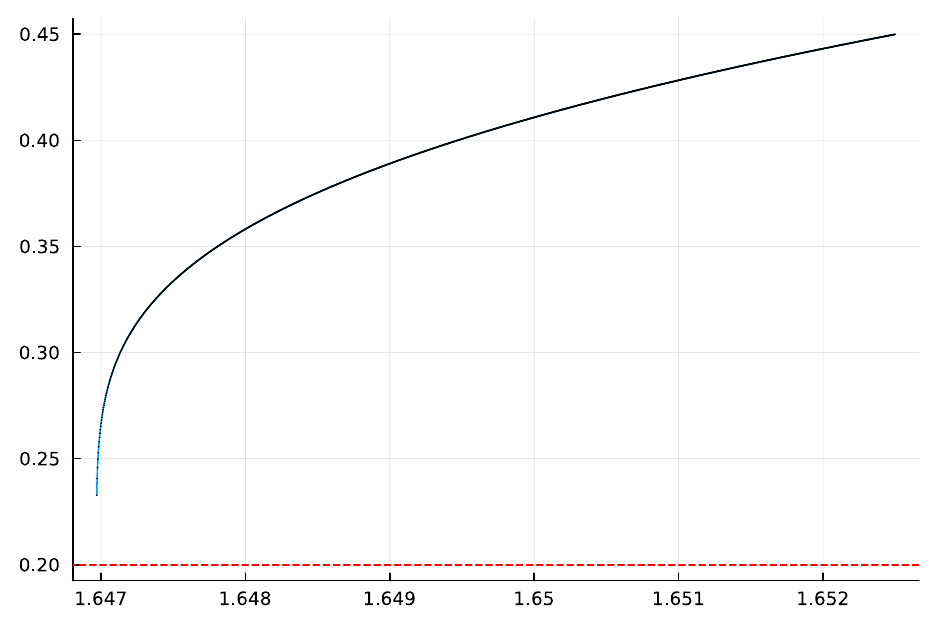}
        \caption{Plot of the minimum value of $\tau(x)$, where $x\in [0,\xi]$, against $\tau_0$. 
 }
        \label{fig:mintfortau}
    \end{subfigure}
    \hfill
    \begin{subfigure}{0.32\textwidth}
        \includegraphics[width=\linewidth]{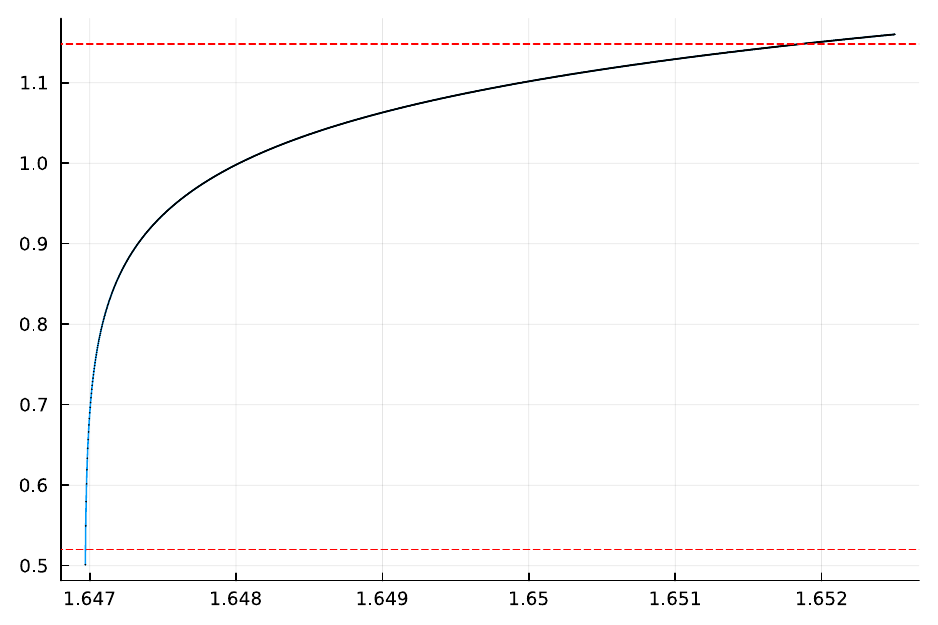}
        \caption{Plot of the deployment angle $\theta=\theta(\tau_0)$ against $\tau_0$.}
        \label{fig:thetafortau}
    \end{subfigure}
    \caption{
Plots of parameters of trajectory $\mathcal T$ as obtained by the solution to the ODE system of Definition~\ref{def: psi tau functions ode and Tau} for initial conditions $\tau_0 \in [1.64697 , 1.6525]$.
    }
    \label{fig: par to ODE system sol}
\end{figure}

We are now ready to conclude with the proof of Theorem~\ref{thm: main optimal result}.

\begin{proof}[Proof of Theorem~\ref{thm: main optimal result}]
By Lemma~\ref{lem: restriction of thetas}, the optimal cost to \adi\ is
$
\inf_{\theta \in [0.52, 1.148]} B_\theta(s(\theta)).
$
By Lemma~\ref{lem: tau ok and theta ok}, all $\tau_0 \in [1.64697 , 1.6525]$ are inspection feasible, and the corresponding deployment angles cover $[0.52, 1.148]$. Thus the inspective curves do not touch the disk, and Theorem~\ref{thm: main result} applies.

Therefore, in order to determine the optimal solution it remains to minimize the cost expression of Theorem~\ref{thm: main result} over the admissible trajectories. The ODE formulation shows that each trajectory is uniquely determined by the initial condition $\tau(0)=\tau_0$, so that the cost becomes a function of this single parameter. Hence the problem reduces to $\spocp(\tau_0)$, namely the problem of minimizing~\eqref{equa: function on xi and t0} over $\tau_0 \in [1.64697 , 1.6525]$.

Since the admissible family is one-dimensional, we evaluate this function numerically and refine the search interval around the minimizer. Figure~\ref{fig: min sol to cont problem} summarizes this coarse-to-fine evaluation over increasingly refined intervals of 2000 grid points each. This refinement brackets the minimizer and provides the required numerical precision.

The minimum is sandwiched between $3.5492598$ and $3.54925986$. For
$\tau_0 = 1.6469768608776936$ (exact value) we obtain
$\xi=0.8119098734258519\ldots$,
$\theta = 0.5909025598581181\ldots$, and a corresponding inspective curve with minimum distance at least $0.0302318\ldots$ to the disk boundary (corresponding to $\tau(\chi)\approx 0.24774522\ldots$ for some $\chi$), and reported cost approximately
$3.5492595860809693\ldots$.
Numerical accuracy of this computation is discussed in Section~\ref{sec:numerics}.
\end{proof}



\begin{figure}[h!]
    \centering
    \begin{subfigure}{0.32\textwidth}
        \includegraphics[width=\linewidth]{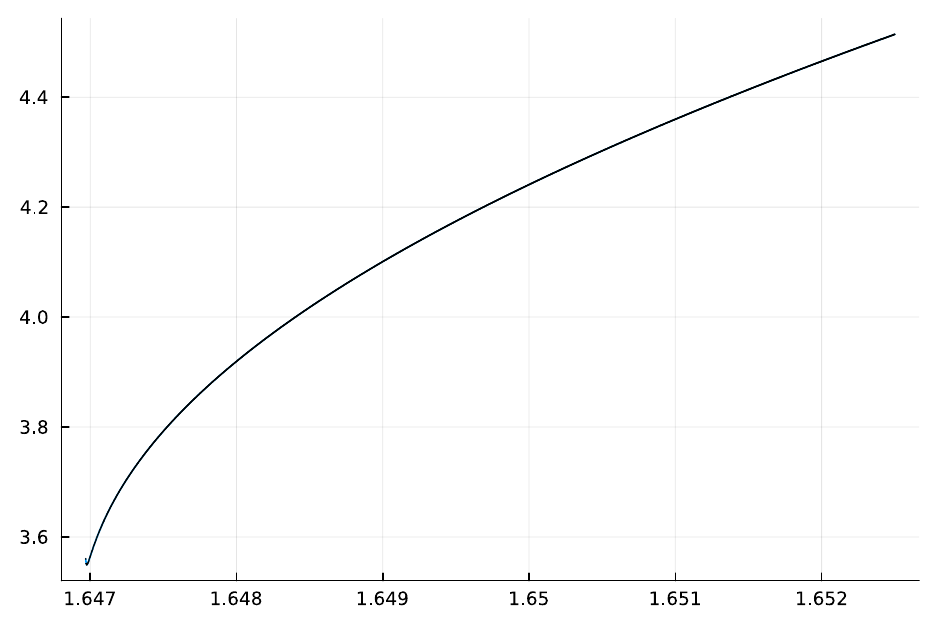}
        \caption{$\tau_0 \in [1.64697 , 1.6525]$}
        \label{fig:alldomain}
    \end{subfigure}
    \hfill
    \begin{subfigure}{0.32\textwidth}
        \includegraphics[width=\linewidth]{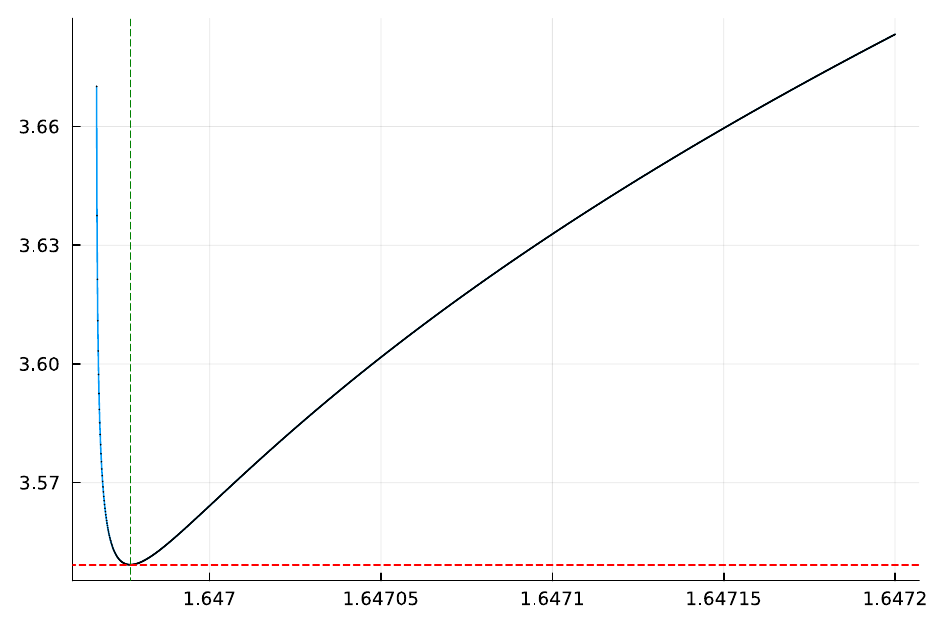}
        \caption{$\tau_0 \in [1.64697 , 1.6472]$}
        \label{fig:closeup}
    \end{subfigure}
    \hfill
    \begin{subfigure}{0.32\textwidth}
        \includegraphics[width=\linewidth]{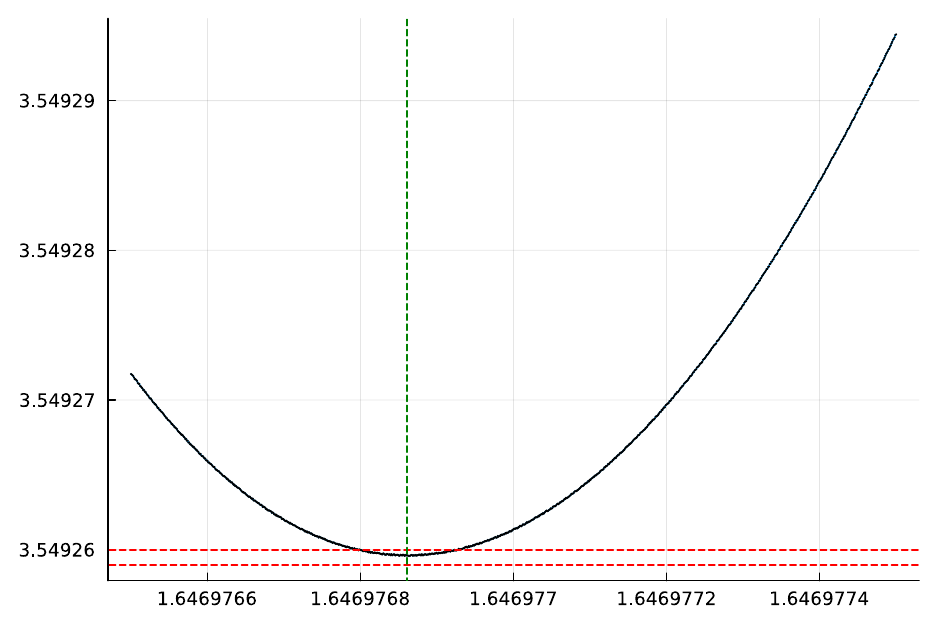}
        \caption{$\tau_0 \in [1.6469764 , 1.6469774]$}
        \label{fig:closerup}
    \end{subfigure}
    \caption{
    Plot of the cost to \adi\ as given by Theorem~\ref{thm: main result} against various initial conditions $\tau(0)=\tau_0$ shown on the horizontal axis. 
    The vertical green dotted line shows $\tau_0=1.64697686$.
    In Figure~\ref{fig:closeup}, the dotted red horizontal line corresponds to value 
$3.5492595$.
    In Figure~\ref{fig:closerup}, the dotted red horizontal lines correspond to values $3.549259$, $3.549260$.     
    }
    \label{fig: min sol to cont problem}
\end{figure}

\ignore{
in powershell run to compress file
& 'C:\Program Files\gs\gs10.05.1\bin\gswin64c.exe' -q -sDEVICE=pdfwrite -sOutputFile='.\\contsolCLOSERRUP_compressed.pdf' '.\\contsolCLOSERRUP.pdf' -dNOPAUSE -dBATCH
}

\ignore{
\begin{figure}[h!]
    \centering
    \begin{subfigure}{0.32\textwidth}
        \includegraphics[width=\linewidth]{figs/thetafortau.pdf}
        \caption{$\tau_0 \in [1.64697 , 1.6525]$}
        \label{fig:alldomain}
    \end{subfigure}
    \hfill
    \begin{subfigure}{0.32\textwidth}
        \includegraphics[width=\linewidth]{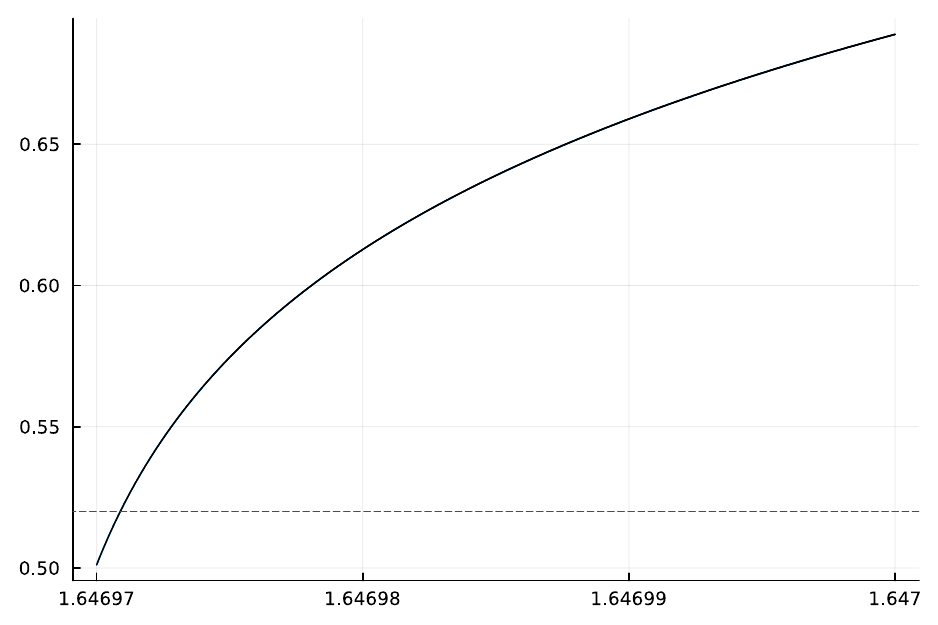}
        \caption{$\tau_0 \in [1.64697 , 1.647]$}
        \label{fig:closeup}
    \end{subfigure}
    \caption{
    Plot of the deployment angle $\theta=\theta(\tau_0)$ against $\tau_0$.
    The red dotted lines correspond to the lines $\theta=0.52$ and $\theta = 1.148$, showing that the range of $\theta(\tau_0)$ is a superset of $[0.52, 1.148]$, when $\tau_0 \in [1.64697 , 1.6525]$. 
    }
    \label{fig: min sol to cont problem}
\end{figure}
}

\subsection{Numerical Methods and Accuracy Guarantees}
\label{sec:numerics}

We used the \texttt{Julia} programming language~\cite{bezanson2017julia} for all computations.

For the convex nonlinear program that lower bounds $s(\theta)$ in Lemma~\ref{lem: theta no small}, we model it with \texttt{JuMP}~\cite{dunning2017jump} and solve it using the interior point method \texttt{Ipopt}~\cite{Ipopt,wachter2006implementation}. The program is convex since the feasible set $\{t_i \geq 0\}$ is convex and the objective is a nonnegative conical combination of norms of affine maps. Hence any KKT point is globally optimal. In practice, the solver returns primal and dual feasible solutions with residuals below $10^{-14}$ and objective values stable to at least $10^{-9}$. Because convexity guarantees global optimality, these certificates validate the solutions and support the digits reported in our lower bounds.

For the ODE system of Definition~\ref{def: psi tau functions ode and Tau}, used in Section~\ref{sec: proof of technical lemma} to establish feasibility of trajectories and to evaluate the cost functional, we proceed as follows. The equation for $\psi$ has a singularity at $x=0$, so we begin the integration at $x_0=10^{-6}$ using the asymptotic expansion $\psi(x_0)=\pi/2-\pi x_0+(\pi^2/2)x_0^2$, whose truncation error is $O(x_0^3)\approx 10^{-18}$ and therefore negligible compared with the solver tolerances. We integrate $\psi$ on $[x_0,1]$ with the adaptive stiff solver \texttt{Rodas5()} from the \texttt{DifferentialEquations.jl} library~\cite{rackauckas2017differential}, setting absolute and relative tolerances to $10^{-12}$, and then integrate $\tau$ on the same interval using the dense output of $\psi$ in the right-hand side. 

The deployment parameter $\xi$ is obtained in two stages. A uniform grid of $10{,}000$ points provides a coarse $\xi_{\text{approx}}$ as the last feasible point according to a geometric predicate that tolerates floating-point noise at the level $10^{-12}$. This value is refined by bisection to absolute tolerance $10^{-8}$. To further confirm stability, we repeat the bisection at tolerance $5\cdot 10^{-10}$ and recompute the objective, reporting the discrepancy between the two runs, which is consistently negligible. In addition, the minimum of $\tau(\cdot)$ on $[x_0,\xi_{\text{approx}}]$ is computed by Brent’s method~\cite{brent2013algorithms} with tolerance $10^{-10}$ to certify that all trajectories remain uniformly outside the disk. 

The integral of Lemma~\ref{lem: sol to cont partial} 
is evaluated using the adaptive Gauss-Kronrod quadrature routine \texttt{quadgk} from \texttt{QuadGK.jl}~\cite{quadgk}, 
with relative tolerance $10^{-12}$ and absolute tolerance $10^{-14}$. Substituting $(\xi,\theta)$ into the expression of Theorem~\ref{thm: main result} then yields the values reported in Section~\ref{sec: proof of technical lemma}.

Each source of numerical error is explicitly controlled. The asymptotic initialization at $x_0=10^{-6}$ contributes error below $10^{-18}$. The ODE solver controls local error to $10^{-12}$. 
The quadrature routine bounds the relative error 
for computing integral of Lemma~\ref{lem: sol to cont partial} to $10^{-12}$. 
The binary search tolerance $10^{-8}$ induces uncertainty on $\theta$ of at most $\pi\cdot 10^{-8}$, with an even tighter self-check available. Finally, the Brent minimization shows that $\tau(x)\geq 0.2$ throughout, implying $\|\mathcal T(x)\|\geq 1.0198$, so the curves remain at least $10^{-2}$ outside the disk. These guarantees confirm that the cost values reported in Theorem~\ref{thm: main optimal result} are reliable to at least six decimal digits.

\section{Discussion}

Bellman introduced the famous lost-in-a-forest problem and proposed several variants almost seventy years ago~\cite{bellman1956minimization}. 
In this work we resolve one of these variants, the Average-Case Disk-Inspection problem. The line of inquiry began with the heuristics of 
Gluss~\cite{gluss1961alternative} in the 1960s and continued through the discretization framework developed recently in~\cite{conley2025multiagentdiskinspectionArxiv,conley2025multi}. 
Our analysis not only establishes the exact optimum with certified numerical accuracy, but also reveals the structural nature of optimal trajectories. 
We show that they arise from a reformulation of the problem as an optics model based on Fermat’s Principle of Least Time, which leads to a single-parameter ODE system. 
Crucially, the resulting optimal trajectories avoid the unit disk, contrary to the conjecture of Gluss.

Beyond closing this specific problem, the methods introduced here suggest a possible direction for approaching other geometric search questions. The reformulation of a many-variable nonconvex program  into a single-parameter optimal control problem shows how optics-inspired principles can reduce  complexity and expose structure that is otherwise hidden. While our techniques were developed for  the disk-inspection setting, the interplay between discrete recursions and continuum limits may find  use in related problems where discretization has been the standard tool but has remained difficult  to analyze or scale.

\section*{Acknowledgements}
We thank Derek Muller and his team at \textit{Veritasium} for an expository discussion of Fermat’s principle of least time, in particular the episode \textit{``The Closest We’ve Come to a Theory of Everything''}, without which the perspective adopted in this work would likely not have emerged.
Special thanks also to Caleb Jones for early discussions of the project.

\bibliographystyle{plain}
\bibliography{BiblioSearch-new}

@InProceedings{conley2025multi,
author="Conley, J.
and Georgiou, K.",
editor="Schmid, Ulrich
and Kuznets, Roman",
title="Multi-agent Disk Inspection",
booktitle="Structural Information and Communication Complexity",
year="2025",
publisher="Springer Nature Switzerland",
address="Cham",
pages="262--280",
}

@article{berzsenyi1995lost,
  title={Lost in a forest (a problem area initiated by the late {Richard E. Bellman})},
  author={Berzsenyi, G.},
  journal={Quantum (November/December, 1995)},
  volume={41},
  year={1995}
}

@article{kubel2021approximation,
  title={On the approximation of shortest escape paths},
  author={K{\"u}bel, D. and Langetepe, E.},
  journal={Computational Geometry},
  volume={93},
  pages={101709},
  year={2021},
  publisher={Elsevier}
}

@misc{gibbs2016bellman,
  title={Bellman’s Escape Problem for Convex Polygons},
  author={Gibbs, P.},
  year={2016},
  publisher={viXra}
}

@article{finch2004lost,
  title={Lost in a forest},
  author={Finch, S. R. and Wetzel, J. E.},
  journal={The American Mathematical Monthly},
  volume={111},
  number={8},
  pages={645--654},
  year={2004},
  publisher={Taylor \& Francis}
}

@article{bezanson2017julia,
  title={Julia: A fresh approach to numerical computing},
  author={Bezanson, J. and Edelman, A. and Karpinski, S. and Shah, V. B.},
  journal={SIAM review},
  volume={59},
  number={1},
  pages={65--98},
  year={2017},
  publisher={SIAM}
}

@article{wachter2006implementation,
  title={On the implementation of an interior-point filter line-search algorithm for large-scale nonlinear programming},
  author={W{\"a}chter, A. and Biegler, L. T.},
  journal={Mathematical programming},
  volume={106},
  number={1},
  pages={25--57},
  year={2006},
  publisher={Springer}
}

@misc{Ipopt,
  author = {{COIN-OR}},
  title = {Ipopt: Interior Point Optimizer},
  howpublished = {\url{https://github.com/coin-or/Ipopt}},
  note = {Accessed: 2024-06-19}
}

@article{dunning2017jump,
  title={JuMP: A modeling language for mathematical optimization},
  author={Dunning, I. and Huchette, J. and Lubin, M.},
  journal={SIAM review},
  volume={59},
  number={2},
  pages={295--320},
  year={2017},
  publisher={SIAM}
}

@InProceedings{AcharjeeGKS19,
  title =	"Lower Bounds for Shoreline Searching With 2 or More Robots",
  author =	"Acharjee, S. and Georgiou, K. and Kundu, S. and Srinivasan, A.",
  publisher =	"Schloss Dagstuhl - LZI",
  year = 	"2019",
  volume =	"153",
  bibdate =	"2020-02-12",
  booktitle =	"23rd OPODIS",
  pages =	"26:1--26:11",
  series =	"LIPIcs",
}

@inproceedings{dobrev2020improved,
  title={Improved lower bounds for shoreline search},
  author={Dobrev, S. and Kr{\'a}lovi{\v{c}}, R. and Pardubsk{\'a}, D.},
  booktitle={International Colloquium on Structural Information and Communication Complexity},
  pages={80--90},
  year={2020},
  organization={Springer}
}

@article{baeza1997searching,
  title={Searching: an algorithmic tour},
  author={Baeza-Yates, R.},
  journal={Encyclopedia of Computer Science and Technology},
  volume={37},
  pages={331--359},
  year={1997}
}

@article{isbell1957optimal,
  title={An optimal search pattern},
  author={Isbell, J. R.},
  journal={Naval Research Logistics Quarterly},
  volume={4},
  number={4},
  pages={357--359},
  year={1957},
  publisher={Wiley Online Library}
}

@article{gluss1961minimax,
  title={The minimax path in a search for a circle in a plane},
  author={Gluss, B.},
  journal={Naval Research Logistics Quarterly},
  volume={8},
  number={4},
  pages={357--360},
  year={1961},
  publisher={Wiley Online Library}
}

@inproceedings{baeza1988searching,
  title={Searching with uncertainty},
  author={Baeza-Yates, R. A. and Culberson, J. C. and Rawlins, G. J. E.},
  booktitle={Scandinavian Workshop on Algorithm Theory},
  pages={176--189},
  year={1988},
  organization={Springer}
}

@article{baeza1995parallel,
  title={Parallel searching in the plane},
  author={Baeza-Yates, R. and Schott, R.},
  journal={Computational Geometry},
  volume={5},
  number={3},
  pages={143--154},
  year={1995},
  publisher={Elsevier}
}

@article{baezayates1993searching,
  title={Searching in the plane},
  author={Baeza-Yates, R. A. and Culberson, J. C. and Rawlins, G. J. E.},
  journal={Information and computation},
  volume={106},
  number={2},
  pages={234--252},
  year={1993},
  publisher={Elsevier}
}

@article{jez2009two,
  title={On the two-dimensional cow search problem},
  author={Je{\.z}, A. and {\L}opusza{\'n}ski, J.},
  journal={Information Processing Letters},
  volume={109},
  number={11},
  pages={543--547},
  year={2009},
  publisher={Elsevier}
}

@article{gal2010search,
  title={Search games},
  author={Gal, S.},
  journal={Wiley Encyclopedia of Operations Research and Management Science},
  year={2010},
  publisher={Wiley Online Library}
}

@article{gluss1961alternative,
  title={An alternative solution to the “lost at sea” problem},
  author={Gluss, B.},
  journal={Naval Research Logistics Quarterly},
  volume={8},
  number={1},
  pages={117--122},
  year={1961},
  publisher={Wiley Online Library}
}

@article{finch2005searching,
  title={Searching for a shoreline},
  author={Finch, S. R. and Zhu, L.-Y.},
  journal={arXiv preprint math/0501123},
  year={2005}
}

@article{langetepe2012searching,
  title={Searching for an axis-parallel shoreline},
  author={Langetepe, E.},
  journal={Theoretical Computer Science},
  volume={447},
  pages={85--99},
  year={2012},
  publisher={Elsevier}
}

@article{pelc2018reaching,
  title={Reaching a Target in the Plane with no Information},
  author={Pelc, A.},
  journal={Information Processing Letters},
  volume={140},
  pages={13--17},
  year={2018},
  publisher={Elsevier}
}

@inproceedings{langetepe2010optimality,
  title={On the optimality of spiral search},
  author={Langetepe, E.},
  booktitle={Proceedings of the twenty-first annual ACM-SIAM symposium on Discrete Algorithms},
  pages={1--12},
  year={2010},
  organization={SIAM}
}

@inproceedings{fricke2016distributed,
  title={A distributed deterministic spiral search algorithm for swarms},
  author={Fricke, G. M. and Hecker, J. P. and Griego, A. D. and Tran, L. T. and Moses, M. E.},
  booktitle={2016 IEEE/RSJ International Conference on Intelligent Robots and Systems (IROS)},
  pages={4430--4436},
  year={2016},
  organization={IEEE}
}

@inproceedings{bouchard2018deterministic,
  title={Deterministic treasure hunt in the plane with angular hints},
  author={Bouchard, S. and Dieudonn{\'e}, Y. and Pelc, A. and Petit, F.},
  booktitle={29th International Symposium on Algorithms and Computation, ISAAC 2018},
  volume={123},
  pages={48--1},
  year={2018},
  organization={Schloss Dagstuhl--Leibniz-Zentrum fuer Informatik}
}

@article{emek2015many,
  title={How many ants does it take to find the food?},
  author={Emek, Y. and Langner, T. and Stolz, D. and Uitto, J. and Wattenhofer, R.},
  journal={Theoretical Computer Science},
  volume={608},
  pages={255--267},
  year={2015},
  publisher={Elsevier}
}

@InProceedings{LangnerKUW15,
  title =	"Overcoming Obstacles with Ants",
  author =	"Langner, T. and Keller, B. and Uitto, J. and Wattenhofer, R.",
  publisher =	"Schloss Dagstuhl - Leibniz-Zentrum fuer Informatik",
  year = 	"2015",
  volume =	"46",
  bibdate =	"2018-08-23",
  booktitle =	"International Conference on Principles of Distributed Systems {(OPODIS)}",
  editor =	"Anceaume, E. and Cachin, C. and Potop-Butucaru, M. G.",
  ISBN = 	"978-3-939897-98-9",
  pages =	"9:1--9:17",
  series =	"LIPIcs",
}

@article{pelc2018information,
  title={Information Complexity of Treasure Hunt in Geometric Terrains},
  author={Pelc, A. and Yadav, R. N.},
  journal={arXiv preprint arXiv:1811.06823},
  year={2018}
}

@article{pelc2019cost,
  title={Cost vs. Information Tradeoffs for Treasure Hunt in the Plane},
  author={Pelc, A. and Yadav, R. N.},
  journal={arXiv preprint arXiv:1902.06090},
  year={2019}
}

@book{alpern2006theory,
  title={The theory of search games and rendezvous},
  author={Alpern, S. and Gal, S.},
  volume={55},
  year={2006},
  publisher={Springer Science \& Business Media}
}

@book{alpern2013search,
  title={Search theory},
  author={Alpern, S. and Fokkink, R. and Gasieniec, L. and Lindelauf, R. and Subrahmanian, V. S.},
  year={2013},
  publisher={Springer}
}

@incollection{CGK19search,
  author      = "Czyzowicz, J. and Georgiou, K. and Kranakis, E.",
  title       = "Group Search and Evacuation",
  editor      = "Flocchini, P. and Prencipe, G. and Santoro, N. ",
  booktitle   = "Distributed Computing by Mobile Entities; Current Research in Moving and Computing",
  publisher   = "Springer",
  year        = 2019,
  pages       = "335-370",
  chapter     = 14,
}

@inproceedings{Emekicalp2014,
  year={2014},
  booktitle={Proceedings of International Colloquium on Automata, Languages, and Programming {(ICALP)}, LNCS 8573},
  title={Solving the ANTS Problem with Asynchronous Finite State Machines},
  author={Emek, Y. and Langner, T. and Uitto, J. and Wattenhofer, R.},
  pages={471-482}
}

@article{bellman1956minimization,
  title={Minimization problem},
  author={Bellman, R.},
  journal={Bull. Amer. Math. Soc},
  volume={62},
  number={3},
  pages={270},
  year={1956}
}

@book{ahlswede1987search,
  title={Search problems},
  author={Ahlswede, R. and Wegener, I.},
  year={1987},
  publisher={John Wiley \& Sons, Inc.}
}

@article{rackauckas2017differential,
  title={Differentialequations. jl--a performant and feature-rich ecosystem for solving differential equations in julia},
  author={Rackauckas, C. and Nie, Q.},
  journal={Journal of open research software},
  volume={5},
  number={1},
  pages={15--15},
  year={2017}
}

@misc{quadgk,
  title = {{QuadGK.jl}: {G}auss--{K}ronrod integration in {J}ulia},
  author = {Steven G. Johnson},
  year = {2013},
  howpublished = {\url{https://github.com/JuliaMath/QuadGK.jl}}
}

@book{brent2013algorithms,
  title={Algorithms for minimization without derivatives},
  author={Brent, R. P.},
  year={2013},
  publisher={Courier Corporation}
}

@misc{georgiou2025spiralsLATIN,
  title={Spirals and Beyond: Competitive Plane Search with Multi-Speed Agents},
  author={Georgiou, K. and Jones, C. and Madej, M.},
  howpublished={arXiv:2508.10793},
  note={Accepted to LATIN 2026},
  year={2025}
}

@article{conley2025multiagentdiskinspectionArxiv,
      title={Disk and Partial Disk Inspection: Worst- to Average-Case and {Pareto} Upper Bounds}, 
      author={Conley, J. and Georgiou, K.},
      year={2025},
   journal={arXiv preprint arXiv:2411.15391},
}

\appendix

\section{Proofs Omitted from Section~\ref{sec: fermat}}
\label{sec: Proofs Omitted from Section sec: fermat}

\begin{proof}[Proof of Lemma~\ref{lem: snell optimal}]
Without loss of generality let $\ell$ be the $x$-axis, with $M_1=\{(x,y):y\ge0\}$ and $M_2=\{(x,y):y\le0\}$. Take $A_1=(a_1,b_1)$ with $b_1>0$ in $M_1$ and $A_2=(a_2,-b_2)$ with $b_2>0$ in $M_2$. For $L=(x,0)\in\ell$ the travel time is
\[
T(x)=\frac{\|A_1-L\|}{s_1}+\frac{\|L-A_2\|}{s_2}
= \frac{\sqrt{(x-a_1)^2+b_1^2}}{s_1}+\frac{\sqrt{(x-a_2)^2+b_2^2}}{s_2}.
\]

For any fixed $L$, the shortest path from $A_1$ to $L$ in $M_1$ is the straight segment $A_1L$. 
Similarly, the shortest path from $L$ to $A_2$ in $M_2$ is $LA_2$. Replacing any detours by these straight segments never increases travel time. Thus there is an optimal path of the form $A_1LA_2$ with a single crossing of $\ell$. The problem reduces to minimizing $T(x)$.

Differentiating $T(x)$ gives
\[
T'(x)=\frac{x-a_1}{s_1\sqrt{(x-a_1)^2+b_1^2}}+\frac{x-a_2}{s_2\sqrt{(x-a_2)^2+b_2^2}},
\]
\[
T''(x)=\frac{b_1^2}{s_1\big((x-a_1)^2+b_1^2\big)^{3/2}}+\frac{b_2^2}{s_2\big((x-a_2)^2+b_2^2\big)^{3/2}}>0,
\]
so $T$ is strictly convex and has a unique minimizer $x^\ast$. The condition $T'(x^\ast)=0$ yields
\[
\frac{|x^\ast-a_1|}{s_1\|A_1L^\ast\|}=\frac{|x^\ast-a_2|}{s_2\|A_2L^\ast\|}.
\]
Defining $\alpha_1,\alpha_2$ as the angles of $A_1L^\ast,L^\ast A_2$, respectively, as in Figure~\ref{fig: snelllaw}, this condition becomes
\[
\frac{\sin\alpha_1}{s_1}=\frac{\sin\alpha_2}{s_2},
\]
which is equivalent to Snell's Law. The strict convexity ensures uniqueness. Hence $A_1\rightarrow L^\ast \rightarrow A_2$ is the unique optimal trajectory.
\end{proof}

\section{Proofs Omitted from Section~\ref{sec: new technical contributions} }
\label{sec: Proofs Omitted from Section sec: new technical contributions}

\begin{lemma}[Well-posedness of $\syst(\tau_0)$]
\label{lem: wellposed}
For every $\tau_0 \in \reals_+$ there exists a unique solution $(\psi,\tau)$ to $\syst(\tau_0)$ with
$\psi,\tau \in C^1((0,1])$ and continuous at $0$, satisfying $\psi(0)=\tfrac{\pi}{2}$ and $\tau(0)=\tau_0$.
Moreover, as $x\rightarrow 0^+$ we have
$$
\psi(x)=\tfrac{\pi}{2}-\pi x + O(x^2)
\quad\text{and}\quad
\tau(x)=\tau_0 - 2\pi x + O(x^2).
$$
\end{lemma}

\begin{proof}[Proof sketch]
Set $\varepsilon(x)=\tfrac{\pi}{2}-\psi(x)$. Since $\cot\psi=\tan\varepsilon$, the $\psi$-equation becomes
\begin{equation*}
x \varepsilon'(x) + \tan\varepsilon(x) = 2\pi x, \qquad \varepsilon(0)=0 .
\end{equation*}
Write $\tan\varepsilon=\varepsilon+r(\varepsilon)$ with $r(\varepsilon)=O(\varepsilon^3)$ as $\varepsilon\to 0$. Then
\begin{equation*}
(x\varepsilon(x))' = 2\pi x - r(\varepsilon(x)) .
\end{equation*}
We show next that this integral form admits a unique continuous solution near $x=0$ by a contraction mapping argument. 
Indeed, integrating from \(0\) to \(x\) gives the fixed point form
$$
\varepsilon(x)=\pi x-\frac{1}{x}\int_0^x r(\varepsilon(u))\,\d u
\qquad (x\in(0,\delta]).
$$
Fix \(a>0\) and consider the complete metric space
$$
X=\{f\in C([0,\delta]) : f(0)=0 \text{ and } \|f\|_\infty \le a\}
$$
with the sup norm. Define \(T:X\to C([0,\delta])\) by
$(Tf)(0)=0$
and
$(Tf)(x)=\pi x-\frac{1}{x}\int_0^x r(f(u))\,\d u \quad (x\in(0,\delta]).
$
Since \(r(\varepsilon)=\tan\varepsilon-\varepsilon\) is \(C^1\) near \(0\) with
\(r(0)=r'(0)=0\), there exists \(C>0\) such that for all \(|y|,|z|\le a\),
$$
|r(y)|\le C|y|^3
\quad\text{and}\quad
|r(y)-r(z)|\le C a^2 |y-z|.
$$
Choosing \(a>0\) so that \(C a^2\le \tfrac12\) and \(\delta\le \tfrac{a}{2\pi}\)
ensures that \(T(X)\subseteq X\) and that \(T\) is a contraction on \(X\).
By the Banach fixed point theorem, \(T\) has a unique fixed point
\(\varepsilon\in X\), yielding a unique continuous solution near \(x=0\).

Expanding the fixed point equation once gives
$\varepsilon(x)=\pi x + O(x^3)$, hence
\begin{equation*}
\psi(x)=\tfrac{\pi}{2}-\pi x + O(x^3),
\end{equation*}
which yields the stated $O(x^2)$ bound in Lemma~\ref{lem: wellposed}.

Given $\psi$, the function $\tau$ satisfies the first-order inhomogeneous ODE
\begin{equation*}
\tau'(x) - 2\pi \cot\psi(x) \tau(x) = -2\pi, \qquad \tau(0)=\tau_0.
\end{equation*}
Since $\cot\psi(x)=\tan\varepsilon(x)=\pi x + O(x^3)$ near $0$, the coefficient is continuous there. Standard existence and uniqueness theorems apply, and a short expansion yields
\begin{equation*}
\tau(x)=\tau_0 - 2\pi x + O(x^2).
\end{equation*}

On $(0,1]$ the right-hand sides are continuous and locally Lipschitz whenever $\psi\notin \pi\mathbb Z$. From
$\psi(x)=\tfrac{\pi}{2}-\pi x + O(x^3)$ near $0$, the trajectory stays away from $\pi\mathbb Z$ on a small interval, and standard continuation extends the unique solution to $[0,1]$.
\end{proof}

\section{Proofs Omitted from Section~\ref{sec: optimal control problem}}
\label{sec: proofs omitted from sec: optimal control problem}

\begin{proof}[Proof of Lemma~\ref{lem: continuum-limit}]
All $O(\cdot)$ estimates below hold with constants that are independent of the index $i$, 
as long as $x_i=i/n$ stays in a fixed interval $[\delta,1]$ away from $0$. 
Near $x=0$, the ODE can be written as
$$
x \frac{d}{dx}(\coss\psi)+\coss\psi=2\pi x\,\sinn\psi,
$$
which shows $\coss{\psi(x)}=\pi x+O(x^2)$. 
Hence the singularity at $0$ is removable and the solution is uniquely determined by $\psi(0)=\pi/2$.

For each $n\ge 1$ set
$$
\alpha=\tfrac{2\pi}{n},\quad x_i=\tfrac{i}{n},\quad \Delta x=\tfrac{1}{n}=\tfrac{\alpha}{2\pi}.
$$
Assume the sequences $(y_i)_{i=0}^n$ and $(t_i)_{i=0}^n$ satisfy, for $i\ge 1$,
$$
\cos y_i=\frac{i}{i+1}\cos\left(y_{i-1}-\alpha\right),
$$
and
$$
t_i-t_{i-1}=\left(t_{i-1}-\tan\left(\tfrac{\alpha}{2}\right)\right)\frac{\sin y_{i-1}}{\sin\left(y_{i-1}-\alpha\right)}-t_{i-1}-\tan\left(\tfrac{\alpha}{2}\right),
$$
with initial conditions $y_0=\pi/2$ and $t_0=\tau_0$. Define the piecewise linear interpolants $\psi_n,\tau_n:[0,1]\to\mathbb R$ by $\psi_n(x_i)=y_i$ and $\tau_n(x_i)=t_i$ for $i=0,\dots,n$.

From the recurrence for $y_i$ we obtain
$$
\cos(y_i)-\cos(y_{i-1})=\frac{i}{i+1}\cos\left(y_{i-1}-\alpha\right)-\cos(y_{i-1}).
$$
Using $i/(i+1)=1-\tfrac{1}{i+1}$ and the expansion $\cos(y_{i-1}-\alpha)=\cos y_{i-1}+\alpha\sin y_{i-1}+O(\alpha^2)$ we get
$$
\cos(y_i)-\cos(y_{i-1})=\alpha\sin y_{i-1}-\tfrac{1}{i+1}\cos y_{i-1}+O(\alpha^2).
$$
On the other hand,
$$
\cos(y_i)-\cos(y_{i-1})=-\sin y_{i-1} (y_i-y_{i-1})+O\left((y_i-y_{i-1})^2\right).
$$
Combining these expressions and dividing by $\Delta x=1/n$ yields
$$
\frac{y_i-y_{i-1}}{\Delta x}=-2\pi+\frac{\cot(y_{i-1})}{x_i}+O(\Delta x).
$$
Hence any subsequential limit $\psi$ of $(\psi_n)$ satisfies
$$
\psi'(x)=-2\pi+\frac{\cot\psi(x)}{x},\qquad \psi(0)=\tfrac{\pi}{2}.
$$

Turning to the recurrence for $t_i$, we use $\tan(\alpha/2)=\alpha/2+O(\alpha^3)$ and
$$
\frac{\sin y_{i-1}}{\sin\left(y_{i-1}-\alpha\right)}=1+\alpha\cot y_{i-1}+O(\alpha^2),
$$
to obtain
$$
t_i-t_{i-1}=\alpha\left(t_{i-1}\cot y_{i-1}-1\right)+O(\alpha^2).
$$
Dividing by $\Delta x$ gives
$$
\frac{t_i-t_{i-1}}{\Delta x}=2\pi\left(t_{i-1}\cot y_{i-1}-1\right)+O(\Delta x).
$$
Thus any limit $\tau$ of $(\tau_n)$ satisfies
$$
\tfrac{1}{2\pi}\tau'(x)=\tau(x)\cot\psi(x)-1,\qquad \tau(0)=\tau_0.
$$

On every interval $[\delta,1]$ with $\delta>0$, the sequences $(\psi_n)$ and $(\tau_n)$ are uniformly Lipschitz and bounded, hence equicontinuous. By the Arzelà-Ascoli Theorem, subsequences converge uniformly to $(\psi,\tau)$, which solves the system on $(0,1]$. The singularity at $x=0$ is removable because $\psi(0)=\pi/2$ makes $\cot\psi(x)/x$ integrable near $0$. Standard ODE continuation and uniqueness arguments then extend the solution to $[0,1]$ with $\psi(x)\in(0,\pi)$. Since the limit is unique, the entire sequences $(\psi_n,\tau_n)$ converge to $(\psi,\tau)$, which therefore satisfy the claimed ODE system.
\end{proof}

\begin{proof}[Proof of Lemma~\ref{lem: trajectory-representation}]
Set $h_i=i/n$ and $\Delta h=1/n$, with $\alpha=2\pi/n$. Define the piecewise linear interpolant $\tau_n:[0,1]\to\mathbb{R}$ by $\tau_n(h_i)=t_i$ for $i=0,\ldots,n$.

Fix $0<\delta\le \xi'<\xi$. By Definition~\ref{def: psi tau functions ode and Tau} and Lemma~\ref{lem: continuum-limit}, we have $\tau_n\to\tau$ uniformly on $[\delta,\xi']$ and $\tau(0)=\tau_0$.

Introduce the continuous family of tangent lines
$$
L(x,t)=
\begin{pmatrix}
\cos(2\pi x)\\
-\sin(2\pi x)
\end{pmatrix}
+t
\begin{pmatrix}
\sin(2\pi x)\\
\ \cos(2\pi x)
\end{pmatrix},\quad x\in[0,1],\ t\in\mathbb{R}.
$$
The discretization uses equal angular steps $\alpha=2\pi/n$ and a clockwise indexing of tangents. With this convention the tangent direction at step $i$ is the clockwise rotation by angle $2\pi h_i$, that is the angle $-2\pi h_i$. Therefore the discrete tangent line $L_i(\cdot)$ coincides with $L(h_i,\cdot)$ for every $i$, and hence
$$
A_i=L_i(t_i)=L(h_i,t_i).
$$
By Definition~\ref{def: tau trajectory}, the curve $\mathcal T$ satisfies $\mathcal{T}(x)=L(x,\tau(x))$ for all $x\in[0,1]$. Hence, for every $i$ with $h_i\in[\delta,\xi']$,
$$
\lVert A_i-\mathcal{T}(h_i)\rVert
=\lVert L(h_i,t_i)-L(h_i,\tau(h_i))\rVert
\le |t_i-\tau(h_i)| \lVert(\sin(2\pi h_i),-\cos(2\pi h_i))\rVert
=|t_i-\tau(h_i)|.
$$
Since $\tau_n(h_i)=t_i$ and $\tau_n\to\tau$ uniformly on $[\delta,\xi']$, it follows that
$$
\max_{i:~h_i\in[\delta,\xi']}\lVert A_i-\mathcal{T}(h_i)\rVert\to 0.
$$

Let $\widetilde{\mathcal{T}}_n$ be the polygonal path obtained by linear interpolation along the segments $A_{i-1}A_i$ on each interval $[h_{i-1},h_i]$. Let $K\subset\mathbb{R}$ be a compact interval containing $\tau([\delta,\xi'])$ and containing $\tau_n([\delta,\xi'])$ for all sufficiently large $n$ (this exists by uniform convergence). The map $(x,t)\mapsto L(x,t)$ is uniformly continuous on $[\delta,\xi']\times K$. Combining uniform continuity of $L$ on $[\delta,\xi']\times K$ with the convergence at the grid points yields
$$
\sup_{x\in[\delta,\xi']}\lVert \widetilde{\mathcal{T}}_n(x)-\mathcal{T}(x)\rVert \to 0
\quad\text{as } n\to\infty.
$$
Since $0<\delta\le \xi'<\xi$ was arbitrary, this proves the claimed convergence on $[0,\xi)$.

Finally,
$$
\mathcal{T}(0)=L(0,\tau(0))=(1,0)+\tau_0(0,-1)=(1,-\tau_0)=A_0,
$$
so $\mathcal{T}(0)=A_0$.
\end{proof}

\begin{proof}[Proof of Lemma~\ref{lem: deployment parameter}]
By feasibility, $\mathcal T$ intersects $x=1$ at some $\xi\in(0,1]$ with $\mathcal T_2(\xi)\ge 0$. 
The tangent at angle $2\pi\xi$ has equation $x\cos(2\pi\xi)+y\sin(2\pi\xi)=1$. 
At $x=1$ this gives 
\[
\mathcal T_2(\xi)=\frac{1-\cos(2\pi\xi)}{\sin(2\pi\xi)}=\tan(\pi\xi).
\]
Writing $a=\pi\xi$ and $\theta=(1-\xi)\pi=\pi-a$ we get $\tan(\theta)=-\tan(a)$, hence $\mathcal T_2(\xi)=\tann{\theta}$.
\end{proof}

\begin{proof}[Proof of Lemma~\ref{lem: sol to cont partial}]
Fix $n\in\naturals$ and set $\alpha=2\pi/n$. For $i=0,1,\ldots,n$ let $\phi_i$ and $L_i(\cdot)$ be as in \eqref{def:phii} and \eqref{equa: parametric tangent line}, and define $A_i=L_i(t_i)$, where $(t_i)$ and $(y_i)$ satisfy \eqref{eq: x_i}-\eqref{eq: t_i} with $y_0=\pi/2$ and $t_0=\tau_0$, where
$$
d_i:=\norm{A_i-A_{i-1}}, \quad i=1,\ldots,n.
$$
By Lemma~\ref{lem: recursion for discrete theta} we have for every $i\ge 1$
$$
d_i=\bigl(t_{i-1}-\tann{\alpha/2}\bigr)\frac{\sinn{\alpha}}{\sinn{y_{i-1}-\alpha}}.
$$

Let $k=k(n)\in\{1,\ldots,n\}$ be the largest index such that the polygonal trajectory through $(A_i)_{i=0}^n$ remains in the halfspace $x\ge 1$, and set $\xi_n:=k/n$. Under the assumption that $\tau_0$ is inspection-feasible (Definition~\ref{def: inspection feasible}), the convergence in Lemma~\ref{lem: trajectory-representation} applies and the polygonal trajectories converge to the curve $\mathcal T$ of Definition~\ref{def: psi tau functions ode and Tau}. In particular, $\xi_n\to\xi\in(0,1]$ and $A_k\to \mathcal T(\xi)$ with $\mathcal T_1(\xi)=1$ and $\mathcal T_2(\xi)\ge 0$. By Lemma~\ref{lem: deployment parameter}, we also have $\mathcal T_2(\xi)=\tann{\theta}$, where $\theta=(1-\xi)\pi$, which shows that $\mathcal T$ satisfies the $\theta$-\adi\ feasibility requirements.

It remains to compute the average cost. Following Lemma~\ref{lem: cost to discrete}, for the discrete $(\theta,k)$-instance the average cost equals
$$
C_{\theta,k}(t)=\frac{1}{k}\sum_{i=0}^{k-1} i \norm{A_{i}-A_{i-1}}
=\frac{1}{k+1}\sum_{i=1}^{k} i d_i.
$$
We pass to the continuum limit $n\to\infty$ with $k=k(n)$ and $\xi_n=k/n\to\xi$. Put $h_i:=i/n$. Using the recurrences and the smooth limit $(\psi,\tau)$ of the interpolants (Definition~\ref{def: psi tau functions ode and Tau}), we have uniformly for $i\le k$ that
$$
t_{i-1}=\tau(h_i)+O(\alpha),\qquad
y_{i-1}=\psi(h_i)+O(\alpha).
$$

Since $\psi(x)\in(0,\pi)$ on $[0,\xi]$ and $\psi'(x)=0$ only when $\cot\psi(x)=2\pi x$, we have $\psi(x)\ge \arctan(1/(2\pi\xi))>0$ for all $x\in[0,\xi]$, hence $\sinn{\psi(x)}\ge m>0$ for all $x\in[0,\xi]$. Using this lower bound we obtain the uniform expansions
$$
\sinn{\alpha}=\alpha+O(\alpha^3),\quad \tann{\alpha/2}=\alpha/2+O(\alpha^3),\quad 
\frac{1}{\sinn{y_{i-1}-\alpha}}=\frac{1}{\sinn{\psi(h_i)}}+O(\alpha),
$$
valid for all $i\le k$. Therefore
$$
d_i=(t_{i-1}-\tann{\alpha/2})\frac{\sinn{\alpha}}{\sinn{y_{i-1}-\alpha}}
=\alpha\frac{\tau(h_i)}{\sinn{\psi(h_i)}}+O(\alpha^2)\quad \text{uniformly for }i\le k.
$$

It follows that
$$
C_{\theta,k}(t)=\frac{1}{k+1}\sum_{i=1}^{k} i d_i
=\frac{\alpha}{k+1}\sum_{i=1}^{k} i\frac{\tau(h_i)}{\sinn{\psi(h_i)}}+O(1/n),
$$
since
$$
\frac{1}{k+1}\sum_{i=1}^{k} i\alpha^2
=\frac{\alpha^2}{k+1}\cdot\frac{k(k+1)}{2}
=\Theta(k\alpha^2)
=\Theta(k/n^2)
=O(1/n),
$$
and $k/n\to \xi$. Observing that the leading term is a Riemann sum with mesh $1/n$, and using $\alpha=2\pi/n$ and $k/n\to\xi$, we obtain
$$
\frac{\alpha}{k+1}\sum_{i=1}^{k} i\frac{\tau(h_i)}{\sinn{\psi(h_i)}}
=\frac{2\pi}{\xi}\frac{1}{n}\sum_{i=1}^{k}(i/n)\frac{\tau(h_i)}{\sinn{\psi(h_i)}}+o(1),
$$
which converges to $\frac{2\pi}{\xi}\int_{0}^{\xi}\frac{x\tau(x)}{\sinn{\psi(x)}}\dd x$.
\end{proof}

\begin{proof}[Proof of Lemma~\ref{lem: int attains min}]
By inspection feasibility and well-posedness of the ODE, $\psi,\tau$ are continuous on $[0,\xi]$, 
with $\psi(0)=\pi/2$ and $\tau(0)=\tau_0$. Since $\psi(x)\in(0,\pi)$ for all $x\in[0,\xi]$, the function $\sinn{\psi(x)}$ is strictly positive there, and by continuity it attains a positive minimum value $m>0$ on $[0,\xi]$. Hence
$$
f(x):=\frac{x \tau(x)}{\sinn{\psi(x)}}
$$
is continuous on $(0,\xi]$ and $\lim_{x\to 0^+} f(x)=0$. Defining $f(0):=0$ gives $f\in C[0,\xi]$. 
By the fundamental theorem of calculus,
$$
I'(\xi)=2\pi f(\xi)=2\pi \xi \frac{\tau(\xi)}{\sinn{\psi(\xi)}},
$$
as claimed.
\end{proof}

\end{document}